\documentclass[journal,12pt,onecolumn]{IEEEtran}
\usepackage[letterpaper, top=1in, bottom=1in, left=1in, right=1in]{geometry}

\usepackage{setspace}

\usepackage[english]{babel}

\usepackage[colorlinks]{hyperref}
\hypersetup{
     colorlinks=true,
     linkcolor=blue,
     filecolor=blue,
     citecolor = blue,
     urlcolor=cyan,
     }

\usepackage{amsfonts}
\usepackage{amsthm}
\usepackage{booktabs}
\usepackage{caption}
\usepackage{cite}
\usepackage{colortbl}
\usepackage{dsfont}
\usepackage{filecontents}
\usepackage{gnuplottex}
\usepackage{mathtools}
\usepackage{microtype}
\usepackage{pgfplots}
\usepackage{subcaption}
\usepackage{dirtytalk}

\usepackage{algorithm}
\usepackage[noend]{algpseudocode}

\usepackage{thmtools, thm-restate}

\pgfplotsset{compat=1.13}

\newtheorem{theorem}{Theorem}

\newtheorem{manualtheoreminner}{Theorem}
\newenvironment{manualtheorem}[1]{%
  \renewcommand\themanualtheoreminner{#1}%
  \manualtheoreminner
}{\endmanualtheoreminner}

\newtheorem{lemma}{Lemma}

\newtheorem{manuallemmainner}{Lemma}
\newenvironment{manuallemma}[1]{%
  \renewcommand\themanuallemmainner{#1}%
  \manuallemmainner
}{\endmanuallemmainner}

\newtheorem{manualpropositioninner}{Proposition}
\newenvironment{manualproposition}[1]{%
  \renewcommand\themanualpropositioninner{#1}%
  \manualpropositioninner
}{\endmanualpropositioninner}

\newtheorem{proposition}{Proposition}

\theoremstyle{definition}
\newtheorem{example}{Example}
\newtheorem{definition}{Definition}
\newtheorem{corollary}{Corollary}

\theoremstyle{definition}
\newtheorem{remark}{Remark}

\theoremstyle{definition}

\DeclareMathOperator*{\argmin}{arg\,min}

\definecolor{DarkGreen}{rgb}{0.0, 0.4, 0.0}
\definecolor{DarkRed}{rgb}{0.5,0.1,0.1}
\definecolor{DarkBlue}{rgb}{0.1,0.1,0.5}
\definecolor{DarkPurple}{rgb}{0.5,0.2,0.5}
\definecolor{DarkTurquoise}{rgb}{0.1,0.5,0.5}

\newcommand{\rev}[1]{\textcolor{black}{#1}}

\newcommand{\ap}{\alpha_{\operatorname{p}}}
\newcommand{\as}{\alpha_{\operatorname{s}}}
\newcommand{\bp}{\beta_{\operatorname{p}}}
\newcommand{\bs}{\beta_{\operatorname{s}}}
\newcommand{\SET}{\operatorname{Set}}
\newcommand{\N}{\operatorname{N}}
\newcommand{\A}{\mathcal{A}}
\newcommand{\squeeze}{\operatorname{squeeze}}
\newcommand{\normal}{\operatorname{normal}}
\newcommand{\negate}{\operatorname{negate}}
\newcommand{\canonical}{\operatorname{canonical}}

\newcommand{\GASP}[1]{$\mathsf{GASP}_{#1}$}

\ifCLASSINFOpdf
\else
\fi
\begin{document}
%
\title{Degree Tables for Secure \\ Distributed Matrix Multiplication}
%
%
%

\author{Rafael G. L. D'Oliveira, Salim El Rouayheb, Daniel Heinlein, and David Karpuk \\ RLE, Massachusetts Institute of Technology, USA\\ ECE, Rutgers University, USA \\ Department of Communications and Networking, Aalto University, Finland \\ AICE, F-Secure Corporation, Finland \\ Emails: rafaeld@mit.edu, salim.elrouayheb@rutgers.edu, \\ daniel.heinlein@aalto.fi, david.karpuk@f-secure.com
\thanks{
  This paper was presented in part in ``Degree Tables for Secure Distributed Matrix Multiplication,'' IEEE Information Theory Workshop (ITW), Visby, Gotland, 2019.
  
  }}

\maketitle


\begin{abstract}
We consider the problem of secure distributed matrix multiplication (SDMM) in which a user wishes to compute the product of two matrices with the assistance of honest but curious servers. We construct polynomial codes for SDMM  by  studying a recently introduced combinatorial tool called the degree table. For a fixed partitioning, minimizing the total communication cost of a polynomial code for SDMM is equivalent to minimizing $N$, the number of distinct elements in the corresponding degree table.

We propose new constructions of degree tables with a low number of distinct elements. These new constructions lead to a general family of  polynomial codes for SDMM, which we call $\mathsf{GASP}_{r}$ (Gap Additive Secure Polynomial codes) parametrized by an integer $r$. $\mathsf{GASP}_{r}$ outperforms all previously known polynomial codes for SDMM under an outer product partitioning. We also present lower bounds on $N$ and prove the optimality or asymptotic optimality of our constructions for certain regimes. Moreover,  we formulate the construction of optimal degree tables as an integer linear program and use it to prove the optimality of $\mathsf{GASP}_{r}$ for all the system parameters that we were able to test.
\end{abstract}


\section{Introduction}

We consider the problem of secure distributed matrix multiplication (SDMM): A user has two matrices, $A \in \mathbb{F}_q^{a \times b}$ and $B \in \mathbb{F}_q^{b \times c}$, and wishes to compute their product, $AB \in \mathbb{F}_q^{a \times c}$, with the assistance of $N$ servers, without leaking any information about either $A$ or $B$ to any server.  We assume that all servers are honest but curious, in that they are not malicious and will faithfully follow the pre-agreed on protocol. However, any $T$ of them may collude to try to deduce information about either $A$ or $B$. This setting was first in introduced in~\cite{ravi2018mmult} with many works following~\cite{Kakar2019OnTC,koreans,d2019gasp,DOliveira2019DegreeTF, Aliasgari2019DistributedAP,aliasgari2020private,Kakar2019UplinkDownlinkTI,Yu2020EntangledPC,mital2020secure,bitar2021adaptive, hasircioglu2021speeding}.

The original performance metric considered was the download cost~\cite{ravi2018mmult}, i.e. the total amount of data downloaded by the user from the servers, with later work considering the total communication cost\cite{Chang2019OnTU,9004505,Kakar2019UplinkDownlinkTI}. In~\cite{jia2019capacity}, the following issue was raised with respect to the SDMM setting: Is it beneficial to offload the computations if security is a concern? Indeed, computing the product $AB$ locally is both secure and has zero communication cost \rev{and is, therefore, optimal from a purely information-theoretic point of view. In \cite{doliveira2020notes}, however, it is shown that if proper coding parameters are chosen, the time needed to perform SDMM is less than for the user to perform the matrix multiplication locally. Hence, information-theoretical converses are not applicable without also taking computational aspects into account.}

For polynomial codes with a fixed matrix partitioning, all these performance metrics are equivalent to minimizing the minimum amount of servers, $N$, needed, often called the recovery threshold of the code. In this paper we consider the outer product partitioning\footnote{In Section \ref{sec: inner vs outer}, we compare the communication costs of the outer product partitioning to the inner product partitioning. We show, for example, that for square matrices, the outer product partitioning achieves a lower asymptotic total communication cost than the inner product partitioning.} given by
\begingroup
\allowdisplaybreaks
\begin{align} \label{partition1}
\begin{aligned}
A = \begin{bmatrix}A_1 \\ \vdots \\ A_K\end{bmatrix} \quad \text{and} \quad
B = \begin{bmatrix}B_1 & \cdots & B_L\end{bmatrix}, \quad
\text{so that} \quad
AB = \begin{bmatrix}
A_1B_1 & \cdots & A_1B_L \\
\vdots & \ddots & \vdots \\
A_KB_1 & \cdots & A_KB_L
\end{bmatrix},
\end{aligned}
\end{align}
\endgroup
where all products $A_kB_\ell$ are well-defined and of the same size.  Computing the product $AB$ is equivalent to computing all the subproducts $A_kB_\ell$.  We then construct a polynomial $h(x)=f(x) \cdot g(x)$, whose coefficients encode the submatrices $A_kB_\ell$, and utilize $N$ servers to compute the evaluations $h(a_1),\ldots,h(a_N)$ for certain $a_1,\ldots,a_N$. The polynomial $h$ is constructed so that no $T$-subset of evaluations reveals any information about $A$ or $B$ ($T$-security), but so that the user can reconstruct all of $AB$ given all $N$ evaluations (decodability). The partition parameters $K$ and $L$ are inversely related to the amount of computation that each of the servers will have to perform. Mathematically, it is convenient to think of the number of servers, $N$, as a function of the partitioning parameters $K$ and $L$, and the security parameter $T$. 

Consider polynomials of the following type:
\begin{align} \label{eq: f and g}
\begin{aligned}
f(x) &= \sum_{k = 1}^KA_kx^{\alpha_k} + \sum_{t = 1}^TR_tx^{\alpha_{K+t}}
\quad \text{and} \quad
g(x) &= \sum_{\ell = 1}^LB_\ell x^{\beta_\ell} + \sum_{t = 1}^TS_t x^{\beta_{L+t}},
\end{aligned}
\end{align}
where $R_t$ and $S_t$ are random matrices used to guarantee privacy. The exponents of the terms in $h(x) = f(x) \cdot g(x)$ will be given by the sum of the exponents, denoted by the vectors $\alpha$ and $\beta$, in $f(x)$ and $g(x)$.

The degree table, first introduced in~\cite{9004505}, of $h(x)$, depicted in Table~\ref{problem}, shows the exponents in $h(x)$ as a function of $\alpha$ and $\beta$. In Theorem~1 of~\cite{9004505}, it is shown that if the degree table satisfies the following conditions: (i)~the numbers in the red block are unique in the table and (ii)~numbers in the green/blue block are pairwise distinct, respectively, then there exists evaluation points such that the polynomial code in Equation~(\ref{eq: f and g}) is decodable and $T$-secure. More so, the minimum number of servers, $N$, is the number of distinct terms in the table. Thus, the main question we are interested in is how to construct a degree table, i.e. $\alpha$ and $\beta$, satisfying conditions (i) and (ii), as to minimize the number of servers, $N$.

\begin{table}
\setlength{\tabcolsep}{3pt}
\centering
\begin{tabular}{c|ccc|ccc}

 & $\beta_1$ & $\!\!\cdots\!\!$ & $\beta_L$ &  \cellcolor{blue!25} $\beta_{L\!+\!1}$ & \cellcolor{blue!25} $\!\!\cdots\!\!$ & \cellcolor{blue!25} $\beta_{L\!+\!T}$  \\
\toprule
$\alpha_1$ & \cellcolor{red!25} $\alpha_1 \!+\! \beta_1$ & \cellcolor{red!25} $\!\!\cdots\!\!$ & \cellcolor{red!25} $\alpha_1 \!+\! \beta_L$ & $\alpha_1 \!+\! \beta_{L\!+\!1}$ & $\!\!\cdots\!\!$ & $\alpha_1 \!+\! \beta_{L\!+\!T}$  \\
$\vdots$ & \cellcolor{red!25} $\vdots$ & \cellcolor{red!25} $\!\!\ddots\!\!$ & \cellcolor{red!25} $\vdots$ & $\vdots$ & $\!\!\ddots\!\!$ & $\vdots$  \\
$\alpha_K$ & \cellcolor{red!25} $\alpha_K \!+\! \beta_1$ & \cellcolor{red!25} $\!\!\cdots\!\!$ & \cellcolor{red!25} $\alpha_K \!+\! \beta_L$ & $\alpha_K \!+\! \beta_{L\!+\!1}$ & $\!\!\cdots\!\!$ & $\alpha_K \!+\! \beta_{L\!+\!T}$  \\
\midrule
\cellcolor{green!25} $\alpha_{K\!+\!1}$ &  $\alpha_{K\!+\!1} \!+\! \beta_1$ & $\!\!\cdots\!\!$ &  $\alpha_{K\!+\!1} \!+\! \beta_L$ & $\alpha_{K\!+\!1} \!+\! \beta_{L\!+\!1}$ & $\!\!\cdots\!\!$ & $\alpha_{K\!+\!1} \!+\! \beta_{L\!+\!T}$ \\
\cellcolor{green!25} $\vdots$ &  $\vdots$ & $\!\!\ddots\!\!$ &  $\vdots$ & $\vdots$ & $\!\!\ddots\!\!$ & $\vdots$ \\
\cellcolor{green!25} $\alpha_{K\!+\!T}$ &  $\alpha_{K\!+\!T} \!+\! \beta_1$ & $\!\!\cdots\!\!$ &  $\alpha_{K\!+\!T} \!+\! \beta_L$ & $\alpha_{K\!+\!T} \!+\! \beta_{L\!+\!1}$ & $\!\!\cdots\!\!$ & $\alpha_{K\!+\!T} \!+\! \beta_{L\!+\!T}$ \\
\bottomrule
\end{tabular}
\caption{The Degree Table. The $\alpha_i$'s and $\beta_i$'s are the exponents of the polynomials $f(x)$ and $g(x)$ in \eqref{eq: f and g} used to  encode  $A$ and $B$, respectively. The table entries are the  monomial  degrees in  $f(x)\cdot g(x)$. The problem is to choose the degrees $\alpha_i$'s and $\beta_i$'s to minimize the number of distinct entries in the table subject to: (i)~Decodability, the numbers in the red block must be unique in the table; (ii)~$T$-security, all numbers in the green/blue block must be pairwise distinct. \rev{In Theorem~1 of \cite{9004505} it is shown that any base field $\mathbb{F}_q$ can be extended in such a way that (i) and (ii) are guaranteed. Thus, in this paper we assume that $\mathbb{F}_q$ already has the required size.}}\label{problem}
\end{table}

\subsection{Related Work}

One  distinguishing factor of the SDMM problem is that both matrices, $A$ and $B$, must be kept secure. In the case where only one of the matrices must be kept secure, one can use methods like Shamir's secret sharing~\cite{sham79} or Staircase codes~\cite{bita17}.

For distributed computations, polynomial codes were originally introduced in~\cite{polycodes1} in a slightly different setting, namely to mitigate stragglers in distributed matrix multiplication. This was followed by a series of works,~\cite{polycodes2,pulkit,pulkit2,fundamental}.

The degree table was introduced in~\cite{9004505}. In this paper, two schemes were presented, $\mathsf{GASP}_{\text{big}}$ and $\mathsf{GASP}_{\text{small}}$. The corresponding $N$ values of $\mathsf{GASP}_{\text{big}}$ and $\mathsf{GASP}_{\text{small}}$ can be computed via formulas presented in~\cite{9004505}. We omit restating them here, since $\mathsf{GASP}_{\text{big}}$ is $\mathsf{GASP}_r$ with $r=\min\{K,T\}$ and $\mathsf{GASP}_{\text{small}}$ is $\mathsf{GASP}_r$ with $r=1$ using the newly introduced common generalization called $\mathsf{GASP}_r$ in Definition~\ref{def:GASPr}. We present a formula to compute the $N$ value of $\mathsf{GASP}_r$ in Theorem~\ref{theo:GASP_r_N}.

The study of the degree table is related to the topic of sumsets in additive combinatorics~\cite{geroldinger2009combinatorial,tao2006additive}. In that context, a problem, usually referred to as the \emph{inverse problem}, is to obtain structural information on two finite sets, $A$ and $B$, given that the cardinality of $A+B=\{a+b : a\in A , b \in B \}$ is small. In our case, the decodability condition implies that at least $KL$ integers appear only with multiplicity one in the sumset and therefore classical theorems about sumsets only yield meaningful results if the product $KL$ is small.

\rev{The literature on SDMM has also studied different variations on the model we focus on here. For instance, in \cite{nodehi2018limited,jia2019capacity,mital2020secure,akbari2021secure} the encoder and decoder are considered to be separate, in \cite{nodehi2018limited} servers are allowed to cooperate, and in \cite{kim2019private} they consider a hybrid between SDMM and private information retrieval where the user has a matrix $A$ and wants to privately multiply it with a matrix $B$ belonging to some public list. The scheme we present here can be readily used or adapted to many of these settings (e.g., \cite{bitar2021adaptive, zhu2021improved}).}

\subsection{Summary of Results}
We summarize below our main results on degree tables and their implications on the constructions of codes for SDMM:

\begin{itemize}
    \item {\bf Constructions:} We construct an improved family of degree tables based   on generalized arithmetic progressions. This leads immediately and transparently to  a new family of polynomial codes for SDMM, which we call \GASP{r} codes.  As far as we know, \GASP{r} codes  achieve the best performance (in terms of number of servers and, hence, of total communication costs) in the literature, so far, for outer product partitioning. \rev{\GASP{r} codes were first introduced in the conference version of this paper \cite{DOliveira2019DegreeTF} together with the formula, without a complete proof, for the recovery threshold.}
    
    \item {\bf Bounds:}   We show in Theorem \ref{theo:lb}  lower bounds on the number of distinct terms in a degree table. This translates into new lower bounds on the number of servers needed for   \GASP{r} codes.  We use these lower bounds   to prove optimality\footnote{The optimality is with respect to the degree table construction, i.e. it is not a general result about SDMM.} or asymptotic optimality of \GASP{r} in certain regimes. \rev{Theorem \ref{theo:lb} first appeared in the conference version of this paper \cite{DOliveira2019DegreeTF}, where only one of the cases was proved.}
    
    \item {\bf Computational approach:}   In Section \ref{sec: linear programming},  we cast the problem of constructing degree tables for any given parameters  as a binary linear program that can be solved using commercially available software such as \texttt{CPLEX}. We accomplish this by first capturing the conditions on a degree table through an infinite number of Boolean variables and constraints.  Then, we bound the value of the entries in the optimal degree table to make the number of the variables and constraints finite. We ran this optimization approach for a large set of parameters and, although we were able to find codes not in the \GASP{r} family which achieve the same performance, we were not able to find any codes that outperform \GASP{r}. 
    

\end{itemize}

\section{Main Results}\label{sec:mainresults}

We begin by introducing our main contribution, $\mathsf{GASP}_{r}$ codes.\footnote{The acronym GASP stands for \say{Gap Additive Secure Polynomial} and was introduced in \cite{d2019gasp}. We utilize the same acronym since the family of codes proposed here, \GASP{r}, is a generalization of those presented in \cite{d2019gasp}.} These codes have four parameters. The partitioning parameters, $K$ and $L$, which are the amount of pieces into which we partition the matrices $A$ and $B$; this determines the amount of work performed by each server.\footnote{Throughout this work, we assume, without loss of generality, that $L\leq K$. In the case where $K<L$ one needs only to interchange the roles of $K$ and $L$ in all the expressions.} The security parameter $T$, which determines the amount of colluding servers which the scheme can tolerate without leaking any information about $A$ or $B$. And the chain length $r$, a parameter which is intrinsic to our code construction and is chosen by the user.

\begin{definition} \label{def:GASPr}
Given the partitioning parameters $K$ and $L$, the security parameter $T$, and the chain length $r$, such that $1\leq r \leq \min\{K,T\}$, we define the polynomial code $\mathsf{GASP}_{r}$ as the polynomials in Equation~\ref{eq: f and g} with exponents $\alpha$ and $\beta$ given by
\begin{itemize}
    \item $\ap = (0,1,\ldots,K-1)$,
    \item $\as = (KL, KL+1, \ldots, KL+r-1, KL+K, KL+K+1, \ldots, KL+K+r-1, \ldots)$ of size $T$,
    \item $\bp = (0,K,\ldots,K(L-1))$, and
    \item $\bs = (KL, KL+1, \ldots, KL+T-1)$,
\end{itemize}
where $L\leq K$, $\alpha$ is the concatenation of $\ap$ and $\as$, and $\beta$ is the concatenation of $\bp$ and $\bs$. The index \say{p} refers to it being the prefix and \say{s} to it being the suffix.

\end{definition}

The following example should make Definition~\ref{def:GASPr} clearer.

\begin{example}
For $K=L=T=4$ we have four $\mathsf{GASP}_{r}$ codes, all of which have the same $\beta = (0,4,8,12,16,17,18,19)$.
\begin{itemize}
    \item For $r=1$: $\alpha = (0,1,2,3,16,20,24,28)$.
    \item For $r=2$: $\alpha = (0,1,2,3,16,17,20,21)$.
    \item For $r=3$: $\alpha = (0,1,2,3,16,17,18,20)$.
    \item For $r=4$: $\alpha = (0,1,2,3,16,17,18,19)$.
\end{itemize}
\end{example}

This family of codes, $\mathsf{GASP}_{r}$, generalizes the codes $\mathsf{GASP}_{\text{small}}$ and $\mathsf{GASP}_{\text{big}}$ first introduced in~\cite{9004505}. Indeed, $\mathsf{GASP}_{\text{small}} = \mathsf{GASP}_{1}$ and $\mathsf{GASP}_{\text{big}} = \mathsf{GASP}_{\min\{K,T\}}$. Table~\ref{tab:KLT4_example} shows the degree tables of all four $\mathsf{GASP}_{r}$ constructions for  $K=L=T=4$. In this case, the chain length which minimizes the amount of distinct terms in the degree table, and thus the recovery threshold, is $r=2$, achieving $N=36$ distinct terms.

\begin{table*}[!t]
\setlength{\tabcolsep}{0.33em}
\begin{subtable}{0.24\textwidth}
\resizebox{\columnwidth}{!}{\begin{tabular}{c|cccc|cccc}
  & 0& 4& 8&12&\cellcolor{blue!25}16&\cellcolor{blue!25}17&\cellcolor{blue!25}18&\cellcolor{blue!25}19\\
\toprule
 0& \cellcolor{red!25}0& \cellcolor{red!25}4& \cellcolor{red!25}8&\cellcolor{red!25}12&16&17&18&19\\
 1& \cellcolor{red!25}1& \cellcolor{red!25}5& \cellcolor{red!25}9&\cellcolor{red!25}13&17&18&19&20\\
 2& \cellcolor{red!25}2& \cellcolor{red!25}6&\cellcolor{red!25}10&\cellcolor{red!25}14&18&19&20&21\\
 3& \cellcolor{red!25}3& \cellcolor{red!25}7&\cellcolor{red!25}11&\cellcolor{red!25}15&19&20&21&22\\
\midrule
\cellcolor{green!25}16&\cellcolor{gray!25}16&\cellcolor{gray!25}20&24&28&32&33&34&35\\
\cellcolor{green!25}20&\cellcolor{gray!25}20&\cellcolor{gray!25}24&\cellcolor{gray!25}28&\cellcolor{gray!25}32&36&37&38&39\\
\cellcolor{green!25}24&\cellcolor{gray!25}24&\cellcolor{gray!25}28&\cellcolor{gray!25}32&\cellcolor{gray!25}36&40&41&42&43\\
\cellcolor{green!25}28&\cellcolor{gray!25}28&\cellcolor{gray!25}32&\cellcolor{gray!25}36&\cellcolor{gray!25}40&44&45&46&47\\
\bottomrule
\end{tabular}}
\caption{$r\!\!=\!1$, $S\!\!=\!14$, $N\!\!=\!41$}\label{tab:KLT4_example_r1}
\end{subtable}
\begin{subtable}{0.24\textwidth}
\resizebox{\columnwidth}{!}{\begin{tabular}{c|cccc|cccc}
  & 0& 4& 8&12&\cellcolor{blue!25}16&\cellcolor{blue!25}17&\cellcolor{blue!25}18&\cellcolor{blue!25}19\\
\toprule
 0& \cellcolor{red!25}0& \cellcolor{red!25}4& \cellcolor{red!25}8&\cellcolor{red!25}12&16&17&18&19\\
 1& \cellcolor{red!25}1& \cellcolor{red!25}5& \cellcolor{red!25}9&\cellcolor{red!25}13&17&18&19&20\\
 2& \cellcolor{red!25}2& \cellcolor{red!25}6&\cellcolor{red!25}10&\cellcolor{red!25}14&18&19&20&21\\
 3& \cellcolor{red!25}3& \cellcolor{red!25}7&\cellcolor{red!25}11&\cellcolor{red!25}15&19&20&21&22\\
\midrule
\cellcolor{green!25}16&\cellcolor{gray!25}16&\cellcolor{gray!25}20&24&28&32&33&34&35\\
\cellcolor{green!25}17&\cellcolor{gray!25}17&\cellcolor{gray!25}21&25&29&\cellcolor{gray!25}33&\cellcolor{gray!25}34&\cellcolor{gray!25}35&36\\
\cellcolor{green!25}20&\cellcolor{gray!25}20&\cellcolor{gray!25}24&\cellcolor{gray!25}28&\cellcolor{gray!25}32&\cellcolor{gray!25}36&37&38&39\\
\cellcolor{green!25}21&\cellcolor{gray!25}21&\cellcolor{gray!25}25&\cellcolor{gray!25}29&\cellcolor{gray!25}33&\cellcolor{gray!25}37&\cellcolor{gray!25}38&\cellcolor{gray!25}39&40\\
\bottomrule
\end{tabular}}
\caption{$r\!\!=\!2$, $S\!\!=\!19$, $N\!\!=\!36$}\label{tab:KLT4_example_r2}
\end{subtable}
\begin{subtable}{0.24\textwidth}
\resizebox{\columnwidth}{!}{\begin{tabular}{c|cccc|cccc}
  & 0& 4& 8&12&\cellcolor{blue!25}16&\cellcolor{blue!25}17&\cellcolor{blue!25}18&\cellcolor{blue!25}19\\
\toprule
 0& \cellcolor{red!25}0& \cellcolor{red!25}4& \cellcolor{red!25}8&\cellcolor{red!25}12&16&17&18&19\\
 1& \cellcolor{red!25}1& \cellcolor{red!25}5& \cellcolor{red!25}9&\cellcolor{red!25}13&17&18&19&20\\
 2& \cellcolor{red!25}2& \cellcolor{red!25}6&\cellcolor{red!25}10&\cellcolor{red!25}14&18&19&20&21\\
 3& \cellcolor{red!25}3& \cellcolor{red!25}7&\cellcolor{red!25}11&\cellcolor{red!25}15&19&20&21&22\\
\midrule
\cellcolor{green!25}16&\cellcolor{gray!25}16&\cellcolor{gray!25}20&24&28&32&33&34&35\\
\cellcolor{green!25}17&\cellcolor{gray!25}17&\cellcolor{gray!25}21&25&29&\cellcolor{gray!25}33&\cellcolor{gray!25}34&\cellcolor{gray!25}35&36\\
\cellcolor{green!25}18&\cellcolor{gray!25}18&\cellcolor{gray!25}22&26&30&\cellcolor{gray!25}34&\cellcolor{gray!25}35&\cellcolor{gray!25}36&37\\
\cellcolor{green!25}20&\cellcolor{gray!25}20&\cellcolor{gray!25}24&\cellcolor{gray!25}28&\cellcolor{gray!25}32&\cellcolor{gray!25}36&\cellcolor{gray!25}37&38&39\\
\bottomrule
\end{tabular}}
\caption{$r\!\!=\!3$, $S\!\!=\!18$, $N\!\!=\!37$}\label{tab:KLT4_example_r3}
\end{subtable}
\begin{subtable}{0.24\textwidth}
\resizebox{\columnwidth}{!}{\begin{tabular}{c|cccc|cccc}
  & 0& 4& 8&12&\cellcolor{blue!25}16&\cellcolor{blue!25}17&\cellcolor{blue!25}18&\cellcolor{blue!25}19\\
\toprule
 0& \cellcolor{red!25}0& \cellcolor{red!25}4& \cellcolor{red!25}8&\cellcolor{red!25}12&16&17&18&19\\
 1& \cellcolor{red!25}1& \cellcolor{red!25}5& \cellcolor{red!25}9&\cellcolor{red!25}13&17&18&19&20\\
 2& \cellcolor{red!25}2& \cellcolor{red!25}6&\cellcolor{red!25}10&\cellcolor{red!25}14&18&19&20&21\\
 3& \cellcolor{red!25}3& \cellcolor{red!25}7&\cellcolor{red!25}11&\cellcolor{red!25}15&19&20&21&22\\
\midrule
\cellcolor{green!25}16&\cellcolor{gray!25}16&\cellcolor{gray!25}20&24&28&32&33&34&35\\
\cellcolor{green!25}17&\cellcolor{gray!25}17&\cellcolor{gray!25}21&25&29&\cellcolor{gray!25}33&\cellcolor{gray!25}34&\cellcolor{gray!25}35&36\\
\cellcolor{green!25}18&\cellcolor{gray!25}18&\cellcolor{gray!25}22&26&30&\cellcolor{gray!25}34&\cellcolor{gray!25}35&\cellcolor{gray!25}36&37\\
\cellcolor{green!25}19&\cellcolor{gray!25}19&23&27&31&\cellcolor{gray!25}35&\cellcolor{gray!25}36&\cellcolor{gray!25}37&38\\
\bottomrule
\end{tabular}}
\caption{$r\!\!=\!4$, $S\!\!=\!16$, $N\!\!=\!39$}\label{tab:KLT4_example_r4}
\end{subtable}
\caption{The degree tables for $\mathsf{GASP}_r$, for all $r$, when $K=L=T=2^2$. As per Proposition~\ref{lem:GASP_r_opt_r2}, the optimal chain length is $r^*=2$, achieving $N=36$. In this case, the best known lower bound is given by Inequality~2 in Theorem~\ref{theo:lb}, which is $N\geq28$. The gray region in the lower half of a degree table consists of the terms which have already appeared before. The number of terms in the gray region is precisely the score, $S$, appearing in the proof of Theorem~\ref{theo:GASP_r_N}.}
\label{tab:KLT4_example}
\end{table*}

In the following theorem we present a closed, though quite big, expression for the number of distinct entries in the degree table of $\mathsf{GASP}_{r}$. As stated previously, this is equivalent to the minimum number of servers needed for the SDMM scheme.

\begin{restatable}{theorem}{theoremgasp} \label{theo:GASP_r_N}
The number of distinct entries in the degree table of $\mathsf{GASP}_{r}$ is given by
\begin{multline*}
  N= KL+2K+3T-2 -\max\{K,\varphi\}+(L-2)\max\{0,\min\{r,r-\varphi\}\}+ \lfloor (T-1)/r \rfloor \min\{T-1,K-r\} \\ -\mathbf{1}_{\varphi < r}\Bigg(\min\{0,\mu-r\} 
  +r(T-1-\mu)/K+\frac{-Kx^2 + (-K-2\max\{0,\varphi\}+2T-2)x+T-1-\mu}{2} \\ -\frac{T-1-\mu}{K} \cdot \frac{T-1+\mu}{2}\Bigg),
\end{multline*}
where $\varphi = T-1-KL+2K$, $\mu \equiv T-1 \pmod{K}$ with $0 \le \mu \le K-1$, and $x = \min\left\{ \frac{T-1-\mu}{K} -\mathbf{1}_{\mu=0}, L-3 \right\}$.
\end{restatable}

The key to proving Theorem~\ref{theo:GASP_r_N} is to determine a parameter $S$ which we call the score of the chain length $r$. The score represents the number of repeated terms in the lower half of the degree table; in Table~\ref{tab:KLT4_example} it is the number of terms in the gray region. After determining $S$, the number of distinct terms in the degree table is given by the sum of the number of distinct terms in the upper half of the degree table, $KL+K+T-1$, and the number of distinct terms in the lower half of the degree table, $T(L+T)-S$. Hence, $N = KL+K+T-1+T(L+T)-S$. We show how to find a closed expression for $S$ in Lemma~\ref{lem:GASP_r_LiRi}.

For a fixed set of parameters, $K$, $L$, and $T$, there are $\min\{K,T\}$ possible values for the chain length $r$. We are particularly interested in the chain length with the best performance, i.e. which minimizes the number of distinct terms in the degree table. 

\begin{definition} \label{def: optimal chain length}
Let $K$ and $L$ be the partitioning parameters, $T$ be the security parameter, and $N(r)$, the number of distinct terms in the degree table constructed by $\mathsf{GASP}_r$. The optimal chain length is defined as
$r^{*} = \argmin\limits_{r  \in \{1,\ldots,\min\{K,T\}\}} N(r)$.
\end{definition}

In Figure~\ref{fig: main results a}, we show the performance of all four $\mathsf{GASP}_r$ schemes for $K=L=4$. As shown in the figure, all possible chain lengths can be optimal depending on the security parameter $T$. The optimal chain length can be found by comparing all the $N(r)$ via Theorem~\ref{theo:GASP_r_N}. In Section~\ref{sec: optimal chain length}, we show that this search space can be reduced and that for certain cases, $r^{*}$ has a simple expression. One example of this is given in the following proposition.

\begin{restatable}{proposition}{lemmansquared} \label{lem:GASP_r_opt_r2}
For $K=L=T=n^2$ the optimal chain length is given by $r^{*} = n$, with 
\begin{align} \label{eq:nsquare}
N = \left\{\begin{matrix}
3 & \text{if} \quad n=1 \\ 
n^4+2n^3+2n^2-n-2 & \text{if} \quad n \geq 2
\end{matrix}\right. .
\end{align}
\end{restatable}

In Section~\ref{sec: lower bounds n} we show the following lower bounds for the degree table.

\begin{theorem}\label{theo:lb}
Let $K$ and $L$ be the partitioning parameters, $T$ the security parameter, and $N$ the number of distinct terms in a degree table. Then the following three inequalities hold.
\begin{enumerate}
\item $KL+\max\{K,L\}+2T-1 \le N$.
\item If $3\max\{K,L\}+3T-2 < KL$ or $2 \le K = L$, then $KL+\max\{K,L\}+2T \le N$.
\item $KL+K+L+2T-1-T\min\{K,L,T\} \le N$.
\end{enumerate}
\end{theorem}

No bound in Theorem~\ref{theo:lb} is strictly larger than the other ones. Indeed, Inequality~3 is stronger than Inequality 1 if and only if $T^2 < \min\{K,L\}$.
Inequality~2 is always stronger than Inequality~1 by one if its condition is met and hence Inequality~3 is stronger than Inequality~2 if and only if its condition is met and $T^2+1 < \min\{K,L\}$.

By comparing the bounds in Theorem~\ref{theo:lb} to the number of distinct terms in $\mathsf{GASP}_r$, counted via Theorem~\ref{theo:GASP_r_N}, we show that $\mathsf{GASP}_{r^{*}}$ is optimal whenever one of the three parameters is equal to $1$, i.e. $K=1$, $L=1$ or $T=1$. We can also show that in the setting of Proposition~\ref{lem:GASP_r_opt_r2}, $\mathsf{GASP}_{r^{*}}$ is asymptotically optimal.

\begin{figure*}[!t]
  \begin{subfigure}[b]{0.48\linewidth}
\begin{gnuplot}[terminal=epslatex, terminaloptions={size 9cm,7cm}]
set xlabel '$T=$ Security Level' offset 0,0.5
set xrange [1:10]
set ylabel '$N=$ Number of Servers' offset 1.5,0
set yrange [20:50]
set key at 10, 31
plot\
    'gasp1.txt' with linespoints title '$\mathsf{GASP}_1$' linewidth 5,\
    'gasp2.txt' with linespoints title '$\mathsf{GASP}_2$' linewidth 5,\
    'gasp3.txt' with linespoints title '$\mathsf{GASP}_3$' linewidth 5,\
    'gasp4.txt' with linespoints title '$\mathsf{GASP}_4$' linewidth 5,\
    (2*x+19) title 'lower bound' linewidth 5,\
\end{gnuplot}
\caption{}
\label{fig: main results a}
  \end{subfigure} \hspace{0.13in}
  \begin{subfigure}[b]{0.48\linewidth}
\begin{gnuplot}[terminal=epslatex, terminaloptions={size 9cm,7cm}]
set xlabel '$n$' offset 0,0.5
set ylabel 'Normalized Number of Servers' offset 1.5,0
set yrange [0.95:2.05]
set key at 29, 1.8
plot\
    'gnuplotdata.txt' using 1:2 title '$\mathsf{GASP}_{\text{small}}$' with linespoints linewidth 5,\
    (x**4+2*x**3+2*x**2-x-2)/(x**4+3*x**2) title '$\mathsf{GASP}_{r^{*}}$' linewidth 5,\
    'gnuplotdata.txt' using 1:4 title '$\mathsf{GASP}_{\text{big}}$' with linespoints linewidth 5,\
    1 title 'lower bound' linewidth 5
\end{gnuplot}
\caption{}
\label{fig: main results b}
  \end{subfigure} 
  \caption{In the figure on the left we present a plot of $\mathsf{GASP}_{r}$ for $r=1,2,3,4$ where $K=L=4$. The lower bound is Inequality 1 in Theorem~\ref{theo:lb}. In the figure on the right we perform a comparison between $\mathsf{GASP}_{\text{small}}$, $\mathsf{GASP}_{r}$ with $r=r^{*}=n$ due to Proposition~\ref{lem:GASP_r_opt_r2}, and $\mathsf{GASP}_{\text{big}}$ for $K=L=T=n^2$. The term ``lower bound'' refers to the left hand side of Inequality~2 in Theorem~\ref{theo:lb} and the number of servers, $N$, is normalized by this lower bound.}
  \label{fig: main results} 
\end{figure*}

\begin{corollary}\label{cor:asympt optimal}
Let $K=L=T=n^2 \ge 4$. Then, $N\geq n^4+3n^2$. Moreover, $\mathsf{GASP}_{r}$ for $r=n$ is asymptotically optimal and within 38\% of the lower bound.
\end{corollary}
\begin{proof}
Inequality~(2) in Theorem~\ref{theo:lb} is $KL+\max\{K,L\}+2T=n^4+3n^2$ and the fraction of the size of the degree table constructed by $\mathsf{GASP}_{r}$, cf. Proposition~\ref{lem:GASP_r_opt_r2}, divided by the left hand side of Inequality~(2) in Theorem~\ref{theo:lb} is
\begin{align}
\begin{aligned}
&\frac{n^4+2n^3+2n^2-n-2}{n^4+3n^2} \le \frac{n^4+2n^3+2n^2}{n^4} = 1+2n^{-1}+2n^{-2} \in 1+\Theta(n^{-1}),
\end{aligned}
\label{eq:ratio}
\end{align}
i.e. the left hand side is asymptotically optimal and its maximum is $<1.38$ at $n \approx 3$.
\end{proof}


In Figure~\ref{fig: main results b}, we show the performance of $\mathsf{GASP}_{\text{small}}$, $\mathsf{GASP}_{\text{big}}$, and $\mathsf{GASP}_{r^*}$ normalized by the best lower bound in Theorem~\ref{theo:lb} for the case where $K=L=T=n^2$. As we can see, both $\mathsf{GASP}_{\text{small}}$ and $\mathsf{GASP}_{\text{big}}$ converge to around two times the lower bound, while $\mathsf{GASP}_{r^*}$ is asymptotically optimal, as stated in Corollary~\ref{cor:asympt optimal}.

In Section~\ref{sec: linear programming}, we present a binary linear program which, given the partitioning parameters $K$ and $L$, and the security parameter $T$, can compute the minimum number of distinct terms $N$ for any degree table. However, the number of variables and constraints are, in general, infinite. The problem can be made finite if upper bounds for the entries in the degree table are given, as we do in Sections \ref{sec: linear bounds for entries} and \ref{sec: operational bounds for entries}. Thus, finding the best practical degree table is a finite problem. We have performed computational searches for degree tables, and although we have found different degree tables which match the performance of $\mathsf{GASP}_{r^*}$, we have not found any degree table which outperforms it.

\section{Notation}

The main symbols used throughout this paper are shown in Table~\ref{tab: definitions}. 

\begin{table}[H] 
\centering
\begin{tabular}{p{1cm} p{7.5cm} } 
Symbol & Definition \\
\toprule
$N$  & Recovery threshold or number of distinct terms in the degree table.  \\[3pt]
$K$  & Partitioning parameter for $A$. \\[3pt]
$L$  & Partitioning parameter for $B$. \\[3pt]
$T$  & Security parameter, i.e. number of colluding servers. \\[3pt]
$f(\cdot)$  & Polynomial code for $A$. \\[3pt]
$g(\cdot)$  & Polynomial code for $B$. \\[3pt]
$h(\cdot)$  & The product of $f$ and $g$. \\[3pt]
$\alpha$  & Vector of exponents of $f$. \\[3pt]
$\beta$  & Vector of exponents of $g$. \\[3pt]
$r$ & Chain length. \\[3pt]
$r^*$ & Optimal chain length.\\[3pt]
$a,b,c$ & Dimensions of the matrices $A$ and $B$.\\
\bottomrule 
\end{tabular}

\caption{List of main symbols.}
\label{tab: definitions}
\end{table}

We use the symbol ``$\mid$'' for horizontal splitting of vectors, and split the vectors $\alpha=\ap\mid\as$ and $\beta=\bp\mid\bs$ in the degree table (Table~\ref{problem}) into a \emph{prefix} and a \emph{suffix} such that the prefix $\ap$ has size $K$, the prefix $\bp$ has size $L$, and both suffixes, $\as$ and $\bs$, have size $T$.

For $A,B \subseteq \mathbb{Z}$ and $x,y \in \mathbb{Z}$, we use the notation $x+yA+B = \{x+ya+b \mid a\in A \text{ and } b \in B\}$.
We abbreviate $[n] = \{0,\ldots,n\} \subseteq \mathbb{Z}$ for a nonnegative integer $n$ and the set consisting of the elements of an integral vector $c$ is abbreviated as $\SET(c)$.

\section{A Motivating Example: \texorpdfstring{$K=L=T=4$}{K=L=T=4}}\label{sec:motivating_example}

In this example we consider the multiplication of two matrices $A$ and $B$ over a finite field $\mathbb{F}_q$, partitioned as
\[
A = \begin{bmatrix}
A_1 \\ A_2 \\ A_3 \\ A_4
\end{bmatrix} \quad \text{and} \quad
B = \begin{bmatrix}
B_1 & B_2 & B_3 &B_4
\end{bmatrix},
\quad \text{so that,} \quad
AB = \begin{bmatrix}
A_1B_1 & A_1B_2 & A_1B_3 & A_1B_4\\
A_2B_1 & A_2B_2 & A_2B_3 & A_2B_4\\
A_3B_1 & A_3B_2 & A_3B_3 & A_3B_4\\
A_4B_1 & A_4B_2 & A_4B_3 & A_4B_4\\
\end{bmatrix}.
\]

We construct a scheme which computes each term $A_kB_\ell$, and therefore all of $AB$ via polynomial interpolation. The scheme must be private for any $T=4$ servers colluding to infer any information about $A$ or $B$.

Let $R_1, \ldots, R_4$ and $S_1, \ldots, S_4$ be matrices chosen independently and uniformly at random with entries in $\mathbb{F}_q$, of sizes equal to the $A_k$ and $B_\ell$, respectively. Define the polynomials
\begin{align*}
f(x) = A_1x^{\alpha_1} + A_2x^{\alpha_2} + A_3x^{\alpha_3} + A_4x^{\alpha_4} + R_1x^{\alpha_5} + R_2x^{\alpha_6} + R_3x^{\alpha_7} + R_4x^{\alpha_8}
\end{align*}
and
\begin{align*}
g(x) = B_1x^{\beta_1} + B_2x^{\beta_2} + B_3x^{\beta_3} + B_4x^{\beta_4} + S_1x^{\beta_5} + S_2x^{\beta_6} + S_3x^{\beta_7} + S_4x^{\beta_8}.
\end{align*}

We recover the products $A_kB_\ell$ by interpolating the product $h(x) = f(x) \cdot g(x)$. Specifically, for some evaluation points $a_n\in \mathbb{F}_q$, we send $f(a_n)$ and $g(a_n)$ to server $n = 1,\ldots,N$, who then responds with $h(a_n) = f(a_n) \cdot g(a_n)$. These evaluations suffice to interpolate $h(x)$, allowing us to retrieve the coefficients of $h(x)$, which in turn allows us to decode all the $A_kB_\ell$. In~\cite{9004505}, it was shown that if the degree table of $\alpha$ and $\beta$ satisfy the conditions in Table~\ref{problem}, then the number of evaluation points needed, $N$, is equal to the number of distinct terms in the degree table.

In Table~\ref{tab:KLT4_example}, we show the degree tables of $\mathsf{GASP}_r$ for all $r$. The upper half of the degree table coincides for every $r$, and consists of the numbers from $0$ to $KL+K+T-2 = 22$. The gray region in the lower half of the degree table consists of the terms which have already appeared before, by ordering them up to down. The number of terms in the gray region is called the score, denoted by $S$, and is a key element in Theorem~\ref{theo:GASP_r_N}.

We calculate the number of distinct terms in the degree table as follows. As seen previously, the upper half of the degree table has $H_1=KL+K+T-1=23$ distinct terms. The lower half has a total of $H_2=T(L+T) = 32$ terms, $S$ of which appear elsewhere. Thus $N = H_1+H_2-S = 55-S$.

In Theorem~\ref{theo:GASP_r_N} we show how to compute the score, $S$, for any $r$. In general, we can determine the best chain length, $r^{*}$, by computing all $\min\{K,T\}=4$ possibilities for $r$ and choosing the one which maximizes the score $S$. In this case, the optimal chain length is $r^{*} = 2$ which could have also been obtained directly through Proposition~\ref{lem:GASP_r_opt_r2}. Thus, for this case, $\mathsf{GASP}_r$ for $r=r^*=2$ is the best known scheme and requires $N=36$ servers. Using the best known lower bound for this case in Theorem~\ref{theo:lb}, we obtain $N\geq KL+\max\{K,L\}+2T= 28$.

\section{\texorpdfstring{$\mathsf{GASP}_{r}$}{GASPr} Codes}\label{sec:GASPr}

In this section, we find the minimum number of servers, $N$, needed for $\mathsf{GASP}_{r}$ codes as presented in Definition~\ref{def:GASPr}. We begin by showing the formal definition of degree tables, first presented in~\cite{d2019gasp}.

\begin{definition}\label{def:degree table}
Let $K$, $L$, and $T$ be positive integers and let $\ap$, $\as$, $\bp$, $\bs$ be non negative integral vectors of sizes $K$, $T$, $L$, and $T$, respectively. Then, the tuple $(\ap,\as,\bp,\bs)$ is called a degree table, if
\begin{enumerate}
\item all integers in $(\ap|\as)$ are distinct,\footnote{Note that $(\ap|\as)$ distinct is equivalent to $\as$ distinct by an application of (3), cf. Table~\ref{problem}. The same is true for $(\bp|\bs)$.}
\item all integers in $(\bp|\bs)$ are distinct, and
\item for any integer $n$ in $\SET(\ap)+\SET(\bp)$, there is a unique $i$ in $\SET(\ap|\as)$ and a unique $j$ in $\SET(\bp|\bs)$ such that $n = i+j$.
\end{enumerate}

If $K$, $L$, and $T$ are clear from the context, we abbreviate $\alpha=(\ap|\as)$, $\beta=(\bp|\bs)$, and $(\alpha,\beta)=(\ap,\as,\bp,\bs)$. We denote the set of all degree tables $(\ap,\as,\bp,\bs)$ of sizes $K$, $T$, $L$, and $T$ by $\A(K,L,T)$.
\end{definition}

As shown in~\cite{d2019gasp}, each degree table corresponds to a polynomial code for secure distributed matrix multiplication. The code is obtained by equating $\alpha$ and $\beta$ in the degree table with the exponents in Equation~\ref{eq: f and g}. It was also shown that the minimum number of servers, $N$, needed for the scheme is given by the cardinality of the sumset $\SET(\alpha)+\SET(\beta)$. Thus, degree tables are a way of transforming the secure distributed matrix multiplication problem into the combinatorial problem shown in Table~\ref{problem}.

\begin{definition}
Let $K$, $L$, and $T$ be positive integers. For $(\alpha,\beta) \in \A(K,L,T)$ we define $\N(\alpha,\beta)$ as the cardinality of the sumset $\SET(\alpha)+\SET(\beta)$.
$\N(K,L,T)$ is the minimum\footnote{Since $\N(\alpha,\beta)$ is a positive integer, the infimum over $(\alpha,\beta) \in \A(K,L,T)$ is attained.} $N(\alpha,\beta)$ for $(\alpha,\beta)$ in $\A(K,L,T)$.
\end{definition}

\begin{remark}
Note that in $\mathsf{GASP}_r$, the vectors $\ap$, $\bp$ and $\bs$ depend on the parameters $K$, $L$, and $T$ but do not depend on the chain length $r$. What changes as a function of the chain length $r$ is the suffix $\as$. This suffix is a generalized arithmetic progression. Indeed, $\SET(\as)$ is precisely the set consisting of the $T$ smallest elements of the sumset $KL+[r-1]+K\mathbb{Z}_{\ge 0}$. The previous constructions presented in~\cite{9004505}, $\mathsf{GASP}_{\text{small}}$ and $\mathsf{GASP}_{\text{big}}$, are particular cases of $\mathsf{GASP}_r$ given by $\mathsf{GASP}_{\text{small}} = \mathsf{GASP}_1$ and $\mathsf{GASP}_{\text{big}} = \mathsf{GASP}_{\min\{K,T\}}$.
\end{remark}

To find the number of distinct terms, $N$, in the degree table of \GASP{r} we need the following.

\begin{definition}\label{def:left right score}
Let $K$ and $L$ be the partitioning parameters, $T$ be the security parameter, and $r$ be the chain length of the code $\mathsf{GASP}_r$. For $1 \le i \le T$ we define $L_i$ (and $R_i$) to be the set of integers that are in the first $L$ (last $T$) entries of row $K+i$ such that these integers appear in the first $K+i-1$ rows of the degree table constructed by $\mathsf{GASP}_r$. We call the cardinalities, $|L_i|$ and $|R_i|$, the left and right score of the row $K+i$, respectively. The score of row $i$ is defined as $S_i = |L_i|+|R_i|$ and the score of the chain length $r$ is defined as
\begin{align}
S(r) = S = \sum_{i=1}^T S_i = \sum_{i=1}^T (|L_i|+|R_i|).\label{eq_slr}
\end{align}
\end{definition}



In Table~\ref{tab:KLT4_example}, the left scores are represented by the gray regions in the lower left side of the degree table, and the right scores by the gray regions in the lower right side. The score, $S$, is the total number of terms in gray, which consists of the terms which have already appeared above in the degree table. Thus, the number of distinct terms in the degree table is given by
\begin{align}
N(r) = KL+K+T-1+T(L+T)-S(r).\label{eq_NS}
\end{align}
Therefore, to determine $N(r)$ we must determine $S(r)$, which is a function of the left and right scores. We present a closed form for these scores in the following lemma.

\begin{lemma}\label{lem:GASP_r_LiRi}
In the setting of Definition~\ref{def:left right score}, it follows that,
\[
|L_i|=
\begin{cases}
\min\{L,2+\lfloor(T-1-i)/K\rfloor\} & \text{if } 1 \le i \le r \\
L & \text{if } r+1 \le i \le T ,
\end{cases}
\]
and,
\[
|R_i|=
\begin{cases}
\max\{0,K+T-KL-1\} & \text{if } i=1 \\
\max\{0,T-K+r-1\} & \text{if $2 \le i$ and $i \equiv 1 \pmod{r}$} \\
T-1 & \text{if } i \not\equiv 1 \pmod{r}
\end{cases}.
\]
\end{lemma}

\begin{proof}
See the Appendix.
\end{proof}

We conclude this section with a formula for the number of distinct terms, $N$, in the degree table of \GASP{r}.

\theoremgasp*

\begin{proof}
See the Appendix.
\end{proof}

\section{The Optimal Chain Length \texorpdfstring{$r^{*}$}{r*}} \label{sec: optimal chain length}

For a fixed choice of the degree table parameters, $K$, $L$, and $T$, there are $\min\{K,T\}$ possible values for the chain length $r$ of $\mathsf{GASP}_r$. We are interested in determining the chain length which minimizes the number of distinct terms, $N$, in the degree table. We refer to this value as the optimal chain length and denote it by $r^*$, as shown in Definition~\ref{def: optimal chain length}. 

In general, we can compute $N(r)$ for all possible $\min\{K,T\}$ values of $r$ using Theorem~\ref{theo:GASP_r_N}, and then compare the values. In this section, we show how to determine the optimal chain length for certain parameters, and how to reduce the search space in general. We begin by determining the optimal chain length for $K=L=T=n^2$.

\lemmansquared*

\begin{proof}
See the Appendix.
\end{proof}

Next, we solve the case for which $r \le \varphi$.

\begin{lemma}\label{lem:GASP_rlephi}
In the setting of Theorem~\ref{theo:GASP_r_N}, if $r \le \varphi$, then $r^{*}=\min\{K,T,\varphi\}$ minimizes $N(r)$.
\end{lemma}

\begin{proof}
See the Appendix.
\end{proof}

In  light of Lemma~\ref{lem:GASP_rlephi}, we can restrict the task to find an $r^{*}$ to 
\begin{align}
r^{*}_1 \in \argmin_{r \in \{\max\{1,\varphi+1\},\ldots,\min\{K,T\}\}}N(r) \label{eq_restrictedminimizer}
\end{align}
and compare $N(r^{*}_1)$ to $N(r^{*}_2)$, where $r^{*}_2 = \max\{1,\min\{K,T,\varphi\}\}$. In the setting of Equation~\ref{eq_restrictedminimizer}, we have $\varphi +1 \le r$, which in turn allows us to simplify the formula for $N(r)$ in Theorem~\ref{theo:GASP_r_N} to
\begin{align}
\label{eq_N_phismallerr}
\begin{split}
N&(r) =
KL+K+3T-2 +(L-2)\min\{r,r-\varphi\} 
\\&+ \lfloor (T-1)/r \rfloor \min\{T-1,K-r\} -\min\{0,\mu-r\}
\\&-r(T-1-\mu)/K -\frac{-Kx^2 + (-K-2\max\{0,\varphi\}+2T-2)x+T-1-\mu}{2}
\\&+\frac{T-1-\mu}{K} \cdot \frac{T-1+\mu}{2}
\end{split}
\end{align}
so that, by omitting and introducing constants in $r$, we obtain
\begin{align*}
\argmin_{r \in \{\max\{1,\varphi+1\},\ldots,\min\{K,T\}\}}N(r)
=
\argmin_{r \in \{\max\{1,\varphi+1\},\ldots,\min\{K,T\}\}}H(r)
\end{align*}
where
\begin{align}
\label{eq_H}
\begin{split}
H(r) = & (L-2-(T-1-\mu)/K)r +\max\{\mu,r\} + \lfloor (T-1)/r \rfloor \min\{T-1,K-r\}.
\end{split}
\end{align}

This optimization problem is not only not differentiable or continuous, but also not convex, as the following example shows.

\begin{example}\label{example:optimal_r}
Let $L=6$, $K=T=9$.
Then $\varphi = -28$ and $\mu = 8$. Thus, Equation~\ref{eq_H} is given by $H(r) = 4r +\max\{8,r\} + \lfloor 8/r \rfloor (9-r)$ with $r \in \{1,\ldots,9\}$.
This yields the following table.
\begin{center}
  \begin{tabular}{c|lllllllll}
  \toprule
$r$&1&2&3&4&5&6&7&8&9\\

$H(r)$&76&44&32&34&32&35&38&41&45\\
\bottomrule
\end{tabular}  
\end{center}
Thus, the set of minima, $\{3,5\}$, is not convex, implying that $H(r)$ is not convex in general.
\end{example}

Fortunately, $H(r)$ is a piecewise linear function, and the pieces change whenever $\lfloor (T-1)/r \rfloor$ jumps or the $\min$ or $\max$ swap.
The next lemma allows us to compute the regions of $r$ for which $\lfloor (T-1)/r \rfloor$ is constant.

\begin{lemma}\label{lem_H_helper}
Let $r,w,T \in \mathbb{Z}_{\ge 1}$.
Then,
\begin{align*}
w=\left \lfloor \frac{T-1}{r} \right \rfloor
\Leftrightarrow
\left\lfloor \frac{T-1}{w+1} \right\rfloor +1 \le r \le \left\lfloor \frac{T-1}{w} \right\rfloor
\quad \text{and} \quad 0=\left \lfloor \frac{T-1}{r} \right \rfloor
\Leftrightarrow
T \le r .
\end{align*}
\end{lemma}
\begin{proof}
This follows directly from  $x-1 < \lfloor x \rfloor \le x$.
\end{proof}

Thus, $H(r)$ is linear on the interval $\big[ \lfloor (T-1)/(w+1) \rfloor +1,\lfloor (T-1)/w \rfloor \big]$, where $w=\lfloor (T-1)/r \rfloor$, provided that $\{\mu,K-T+1\} \cap \big[ \lfloor (T-1)/(w+1) \rfloor +2,\lfloor (T-1)/w \rfloor-1 \big] = \emptyset$ and we can compute its slope, $s(r) = L-2-(T-1-\mu)/K +\mathbf{1}_{\{\mu < r\}} -w \cdot \mathbf{1}_{\{K-T+1 < r\}}$. The set of all possible $w$ is $ W := \{ \lfloor (T-1)/i \rfloor \mid i \in \{\max\{1,\varphi+1\},\ldots,\min\{K,T-1\}\} \}$.

Abbreviating, $l_w := \lfloor (T-1)/(w+1) \rfloor +1$, $r_w := \lfloor (T-1)/w \rfloor$, and $A_w := \{\mu,K-T+1\} \cap [l_w+1,r_w-1]$, we define for all $w \in W$, where for brevity we apply the rule that if more than one if-clause is valid, the uppermost is taken,
\begin{align*}
Q_w := 
\begin{cases}
\emptyset & \text{ if } r_w < l_w \\
\{l_w\} & \text{ if } l_w = r_w \\
\{l_w\} & \text{ if } A_w \ne \emptyset \text{ and } s(l_w) \ge 0 \text{ and } s(r_w) \ge 0\\
\{l_w,r_w\} & \text{ if } A_w \ne \emptyset \text{ and } s(l_w) \ge 0 \text{ and } s(r_w) < 0\\
A_w & \text{ if } A_w \ne \emptyset \text{ and } s(l_w) < 0 \text{ and } s(r_w) \ge 0\\
\{r_w\} & \text{ if } A_w \ne \emptyset \text{ and } s(l_w) < 0 \text{ and } s(r_w) < 0\\
\{l_w\} & \text{ if } s(l_w) \ge 0 \\
\{r_w\} & \text{ if } s(l_w) < 0 \\
\end{cases}
\end{align*}
and $Q := \cup_{w \in W} Q_w$.

We can now reduce the size of the set containing the minimum of $N(r)$.

\begin{theorem} \label{theo: set of min r}
In the setting of Theorem~\ref{theo:GASP_r_N}, let $Q'  := \{\max\{1,\min\{K,T,\varphi\}\},\max\{1,\varphi+1\},T\}$ and $Q'' := (Q' \cup Q) \cap \{1,\ldots,\min\{K,T\}\}$. Then,
\begin{align*}
Q'' \cap \left( \argmin_{r \in \{1,\ldots,\min\{K,T\}\}}N(r) \right)
\end{align*}
is nonempty, and hence only values in $Q''$ need to be tested to find a global minimum of $N(r)$.
\end{theorem}
\begin{proof}
Following the discussion in this section, we have that $Q$ contains at most one end point of each linear part of the piecewise linear function unless $A_w \ne \emptyset$.
The special cases arising
by Lemma~\ref{lem:GASP_rlephi}, $\{\max\{1,\min\{K,T,\varphi\}\}$ and $\max\{1,\varphi+1\}$,
and
by Lemma~\ref{lem_H_helper}, $T$,
are tested via $Q'$.
The special cases arising by $A_w$, $\mu$ and $K-T+1$, are embedded in the reasoning of $Q$.
\end{proof}

In particular, at most $5+\#W$ values of $r$ must be tested to find a global minimum of $N(r)$, which can be significantly smaller than testing all $\min\{K,T\}$ values in $\{1,\ldots,\min\{K,T\}\}$. Indeed, checking all parameters $1 \le L \le K \le 300$ and $1 \le T \le 300$, we find that the arithmetic mean of $\frac{
5+\#W
}{
\min\{K,T\}
}$ is approximately $32.5\%$.

\begin{example}
Continuing Example~\ref{example:optimal_r}, we have $W=\{1,2,4,8\}$, $Q'=\{1,9\}$, and $Q=\{1,2,3,5\}$. Hence, instead of testing all 9 values in $\{1,\ldots,9\}$, we can restrict our search to the smaller set $Q''=\{1,2,3,5,9\}$. Indeed, $\argmin_{r \in \{1,\ldots,\min\{K,T\}\}}N(r) = \{3,5\}$ and $Q''$ intersect nontrivially.
\end{example}

\section{Equivalence of Degree Tables} \label{sec: equivalence}

In this section, we introduce a framework for dealing with general degree tables. The results presented here serve as tools for proving bounds on the degree table, in Section~\ref{sec: sumsets and lower bounds}, and for analyzing the degree tables as an integer linear programming problem in Section~\ref{sec: linear programming}. The main idea is to define an equivalence relation on $\A(K,L,T)$ under which $\N(\alpha,\beta)$ is invariant.




We use the following additional notation: we denote the all-one vector of length $n$ by $\mathds{1}_n$, the symmetric group on $n$ symbols by $\mathbb{S}_n$, and, for a given permutation $\pi\in\mathbb{S}_n$ and a vector $(v_1,\ldots,v_n)$, we define the operation $\pi\circ(v_1,\ldots,v_n) = (v_{\pi(1)},\ldots,v_{\pi(n)})$. Let $v$ be an $n$-dimensional integral vector. A permutation $\pi \in \mathbb{S}_n$ is called a \emph{sorting permutation} with respect to $v$ if $(\pi \circ v)_i \le (\pi \circ v)_{i+1}$ for all $1 \le i < n$. Note that the sorting permutation is unique if and only if $v$ contains no repeated numbers.

The following lemma presents two operations which can be performed on the vectors $\alpha$ and $\beta$ of a degree table without changing the number of distinct terms $N(\alpha,\beta)$.

\begin{lemma}\label{lem:squeeze}
Let $(\alpha,\beta) \in \A(K,L,T)$ be a degree table and $p=\pi \circ \alpha$, $q=\tau \circ \beta$ the corresponding sorted vectors for unique sorting permutations $\pi$ and $\tau$.
Moreover, we abbreviate $a=\min\{\SET(\alpha)\}$, $A=\max\{\SET(\alpha)\}$, $b=\min\{\SET(\beta)\}$, and $B=\max\{\SET(\beta)\}$. 

\begin{enumerate}
\item For an integer $i$ with $1 \le i < K+T$ and $p_i+B < (p_{i+1}-1)+b$, there is a degree table $(\alpha',\beta) \in \A(K,L,T)$ with $\N(\alpha,\beta) = \N(\alpha',\beta)$ where
$
\alpha'_j
=
\begin{cases}
\alpha_j & \text{if } j \le \pi(i) \\
\alpha_j-1 & \text{else}
\end{cases}
$.
\item For an integer $i$ with $1 \le i < L+T$ and $q_i+A < (q_{i+1}-1)+a$, there is a degree table $(\alpha,\beta') \in \A(K,L,T)$ with $\N(\alpha,\beta) = \N(\alpha,\beta')$ where
$
\beta'_j
=
\begin{cases}
\beta_j & \text{if } j \le \tau(i) \\
\beta_j-1 & \text{else}
\end{cases}
$.
\end{enumerate}

\end{lemma}

\begin{proof}
See the Appendix.
\end{proof}

If one applies Lemma~\ref{lem:squeeze} multiple times, then the ordering of the operations (1.) and (2.) in Lemma~\ref{lem:squeeze} is unique.

\begin{lemma}\label{lem:squeeze_ordering}
Let $(\alpha,\beta)$ be a degree table.
Then at most one of the two operations (1.) and (2.) in Lemma~\ref{lem:squeeze} are applicable for $(\alpha,\beta)$.
Two successive operations of (1.) respectively (2.) using the indices $i,i'$ applied to $(\alpha,\beta)$ yield the same degree table as if the indices $i',i$ were used.
\end{lemma}

\begin{proof}
See the Appendix.
\end{proof}

\begin{lemma}\label{lem:sq_2}
Using the notation of Lemma~\ref{lem:squeeze}, if $p_{i}+B \le p_{i+1}+b$, then operation (2.) cannot be applied and if $q_{i}+A \le q_{i+1}+a$, then operation (1.) cannot be applied.
\end{lemma}
\begin{proof}
We only proof the first part, the second part follows with a similar argument.
Assume operation (ii) would be applicable, i.e., there is a $j$ with $q_{j}+A < (q_{j+1}-1)+a$.
Then,
$q_{j+1}+a
\le B+p_{i}
\le p_{i+1}+b
\le A+q_{j}$
is a contradiction.
\end{proof}

Quite figuratively, we call the operations defined in Lemma~\ref{lem:squeeze} a \emph{squeezing} of the degree table $(\alpha,\beta)$. Using Lemma~\ref{lem:squeeze_ordering}, we denote by $\squeeze(\alpha,\beta,n)$, the unique degree table which is the result of $n$ squeezings of $(\alpha,\beta)$, using the smallest feasible index in each step. We denote by $\squeeze(\alpha,\beta)$, the unique degree table such that there is a non-negative integer $m$ with $\squeeze(\alpha,\beta) = \squeeze(\alpha,\beta,m)$ and such that neither the application of (1.) nor (2.) of Lemma~\ref{lem:squeeze} on $\squeeze(\alpha,\beta)$ is feasible. In particular $\squeeze(\alpha,\beta,0) = (\alpha,\beta)$. If $(\alpha,\beta)=\squeeze(\alpha,\beta)$, we say $(\alpha,\beta)$ is \emph{squeezed}.

\begin{proposition}\label{prop:sq_3}
For any non-squeezed degree table $(\alpha,\beta)$, $\squeeze(\alpha,\beta)$ is built by either applying only operation (1.) of Lemma~\ref{lem:squeeze} possibly multiple times or applying only operation (2.) of Lemma~\ref{lem:squeeze} possibly multiple times.
In both cases, the result of the process in independent of the order of the operations.
\end{proposition}
\begin{proof}
This follows from Lemma~\ref{lem:squeeze_ordering} and Lemma~\ref{lem:sq_2}, if operation (1.) is applicable, i.e., there is an $i$ with $p_i+B < (p_{i+1}-1)+b$, after applying operation (1.) $s$ times with index $i$, we have with $\widetilde{p_{i+1}} := p_{i+1}-s$ then $p_i+B = (\widetilde{p_{i+1}}-1)+b$ and hence operation (2.) cannot be applied by Lemma~\ref{lem:sq_2} for all subsequent applications of operation (1.). The other statement follows similarly.
\end{proof}

\begin{table*}[!t]
\setlength{\tabcolsep}{0.33em}
\begin{subtable}{0.24\textwidth}
\begin{center}
\begin{tabular}{c|cccc}
  & 0& 2& \cellcolor{blue!25}4 & \cellcolor{blue!25}5\\
\toprule
 0 & \cellcolor{red!25}0 & \cellcolor{red!25}2 & 4 & 5\\
 1 & \cellcolor{red!25}1 & \cellcolor{red!25}3 & 5 & 6\\
 \cellcolor{green!25} 9 & 9 & 11 & 13 & 14\\ 
 \cellcolor{green!25} 10 & 10 & 12 & 14 & 15\\
\bottomrule
\end{tabular}    
\end{center}
\caption{}\label{tab:operations 1}
\end{subtable}
\begin{subtable}{0.24\textwidth}
\begin{center}
\begin{tabular}{c|cccc}
  & 0& 2& \cellcolor{blue!25}4 & \cellcolor{blue!25}5\\
\toprule
 0 & \cellcolor{red!25}0 & \cellcolor{red!25}2 & 4 & 5\\
 1 & \cellcolor{red!25}1 & \cellcolor{red!25}3 & 5 & 6\\
 \cellcolor{green!25} 8 & 8 & 10 & 12 & 13\\ 
 \cellcolor{green!25} 9 & 9 & 11 & 13 & 14\\
\bottomrule
\end{tabular}    
\end{center}
\caption{}\label{tab:operations 2}
\end{subtable}
\begin{subtable}{0.24\textwidth}
\begin{center}
\begin{tabular}{c|cccc}
  & 0& 1& \cellcolor{blue!25}8 & \cellcolor{blue!25}9\\
\toprule
 0 & \cellcolor{red!25}0 & \cellcolor{red!25}1 & 8 & 9\\
 2 & \cellcolor{red!25}2 & \cellcolor{red!25}3 & 10 & 11\\
 \cellcolor{green!25} 4 & 4 & 5 & 12 & 13\\ 
 \cellcolor{green!25} 5 & 5 & 6 & 13 & 14\\
\bottomrule
\end{tabular}    
\end{center}
\caption{}\label{tab:operations 3}
\end{subtable}
\begin{subtable}{0.24\textwidth}
\begin{center}
\begin{tabular}{c|cccc}
  & 0& 1& \cellcolor{blue!25}7 & \cellcolor{blue!25}8\\
\toprule
 0 & \cellcolor{red!25}0 & \cellcolor{red!25}1 & 7 & 8\\
 2 & \cellcolor{red!25}2 & \cellcolor{red!25}3 & 9 & 10\\
 \cellcolor{green!25} 4 & 4 & 5 & 11 & 12\\ 
 \cellcolor{green!25} 5 & 5 & 6 & 12 & 13\\
\bottomrule
\end{tabular}    
\end{center}
\caption{}\label{tab:operations 4}
\end{subtable}
\caption{In Table (a) we present an initial degree table. After one operation (1.) from Lemma \ref{lem:squeeze} we obtain Table (b). In Table (c) we present a degree table, which is actually the transpose of Table (b). After one operation (2.) from Lemma \ref{lem:squeeze} we obtain Table (d). Table (d) is a squeezed degree table. As illustrated in this figure, an operation of type (1.) on a degree table is equivalent to an operation of type (2.) on the transpose of that same degree table.}
\label{tab:operations}
\end{table*}

We illustrate the squeezing of a degree table in Table \ref{tab:operations}.

Lemma~\ref{lem:squeeze} allows us to bound the integral entries of a squeezed degree table.

\begin{corollary}\label{cor:finite}
Let $(\alpha,\beta) \in \A(K,L,T)$ be a squeezed degree table.
Using the unique sorting permutations $\pi$ and $\tau$ for $\alpha$ and $\beta$, we have
$(\pi\circ\alpha_{i+1}-1)+\min\{\SET(\beta)\}\le\pi\circ\alpha_i+\max\{\SET(\beta)\}$ for all integers $i$ with $1 \le i < K+T$
and
$(\tau\circ\beta_{i+1}-1)+\min\{\SET(\alpha)\}\le\tau\circ\beta_i+\max\{\SET(\alpha)\}$ for all integers $i$ with $1 \le i < L+T$.

In particular, we have
$\max\{\SET(\alpha)\}-\min\{\SET(\alpha)\} \le (K+T-1)(\max\{\SET(\beta)\}-\min\{\SET(\beta)\}+1)$
and
$\max\{\SET(\beta)\}-\min\{\SET(\beta)\} \le (K+T-1)(\max\{\SET(\alpha)\}-\min\{\SET(\alpha)\}+1)$.
\end{corollary}
\begin{proof}
The first part is obvious.
For the second part:
\begin{align*}
&
\max\{\SET(\alpha)\}-\min\{\SET(\alpha)\}
=
\pi\circ\alpha_{K+T}-\pi\circ\alpha_{1}
=
\sum_{i=1}^{K+T-1} \pi\circ\alpha_{i+1}-\pi\circ\alpha_{i}
\\
&
\le
\sum_{i=1}^{K+T-1} \max\{\SET(\beta)\}-\min\{\SET(\beta)\}+1
=
(K+T-1)(\max\{\SET(\beta)\}-\min\{\SET(\beta)\}+1)
\end{align*}
The last part follows by a similar argument.
\end{proof}

Now, we define an equivalence relation on the space $\A(K,L,T)$.

\begin{definition}\label{def:equivalence}
Let $(\ap,\as,\bp,\bs)$ and $(\ap',\as',\bp',\bs')$ be in $\A(K,L,T)$.
Then $(\ap,\as,\bp,\bs)$ and $(\ap',\as',\bp',\bs')$ are called \emph{equivalent}, i.e. $(\ap,\as,\bp,\bs) \sim (\ap',\as',\bp',\bs')$, if there are rational numbers $r,s,t$ and permutations $\kappa\in\mathbb{S}_K, \lambda\in\mathbb{S}_T, \mu\in\mathbb{S}_L, \nu\in\mathbb{S}_T$ such that
\begin{align*}
&
(\ap,\as,\bp,\bs)
=
\\&
r \cdot (\kappa\circ\ap'+s\mathds{1}_K,\lambda\circ\as'+s\mathds{1}_T,\mu\circ\bp'+t\mathds{1}_L,\nu\circ\bs'+t\mathds{1}_T).
\end{align*}
\end{definition}

Thus, two degree tables with same parameters $K,L,T$ are equivalent, if and only if one can be transformed into the other by first permuting each of the four entries in the tuple separately, then translating the permuted $\alpha$ and $\beta$ by not necessary the same values and finally multiplying the permuted and translated $\alpha$ and $\beta$ by the same value.

We abbreviate $r \cdot (\sigma \circ (\alpha,\beta) + (s,t)) := r \cdot (\kappa\circ\ap+s\mathds{1}_K,\lambda\circ\as+s\mathds{1}_T,\mu\circ\bp+t\mathds{1}_L,\nu\circ\bs+t\mathds{1}_T)$ for a degree table $(\alpha,\beta)=(\ap,\as,\bp,\bs)$, where $\sigma = (\kappa,\lambda,\mu,\nu) \in \mathbb{S}_K \times \mathbb{S}_T \times \mathbb{S}_L \times \mathbb{S}_T$, and omit $r=1$, $\sigma=(id,id,id,id)$ or $(s,t)=(0,0)$.

\begin{lemma} \label{lem:equivalence relation}
The relation defined in Definition~\ref{def:equivalence} is an equivalence relation and it holds that, if $(\ap,\as,\bp,\bs) \sim (\ap',\as',\bp',\bs')$, then $\N(\ap,\as,\bp,\bs)=\N(\ap',\as',\bp',\bs')$.
\end{lemma}

\begin{proof}
See the Appendix.
\end{proof}

Our next goal is to introduce a canonical representative for any given degree table.

\begin{definition}
A degree table $(\ap,\as,\bp,\bs)$ is called \emph{normal}, if the entries of $\ap$, $\as$, $\bp$, $\bs$ are sorted in increasing order, respectively, the smallest entry in $(\ap,\as)$, respectively in $(\bp,\bs)$, is zero, and the greatest common divisor of $\SET(\ap|\as|\bp|\bs)$ is one.

The \emph{negate} degree table of $(\ap,\as,\bp,\bs)$ is $-(\ap,\as,\bp,\bs)$ translated by $\max \SET(\ap|\as|\bp|\bs)$.

A degree table $(\ap,\as,\bp,\bs)$ is called \emph{canonical} if it is normal and its normal negate degree table is lexicographically at least as large as $(\ap,\as,\bp,\bs)$.
\end{definition}

The definition immediately provides a construction for the normal degree table of $(\alpha,\beta)$, denoted by $\normal(\alpha,\beta)$, the negate degree table, denoted by $\negate(\alpha,\beta)$, and the canonical degree table, denoted by $\canonical(\alpha,\beta)$. Note that $(\alpha,\beta) \sim\normal(A)\sim \negate(\alpha,\beta)\sim\canonical(\alpha,\beta)$ and for any two equivalent degree table $(\alpha,\beta)\sim (\alpha',\beta')$, we have $\canonical(\alpha,\beta) = \canonical(\alpha',\beta')$ and $\normal(\alpha,\beta) \in \{\normal(\alpha',\beta'),\normal(\negate(\alpha',\beta'))\}$. Hence, without loss of generality, we can assume that any degree table is normal.

In the special case of $K=L$, it is also possible to augment the equivalence relation defined in Definition~\ref{def:equivalence} with an additional operation called transposition which interchanges $\alpha$ and $\beta$, see Table~\ref{tab: optimal} for an example.

\begin{table*}[!t]
\setlength{\tabcolsep}{0.33em}
\begin{subtable}{0.24\textwidth}
\begin{center}
\begin{tabular}{c|ccc}
  & 2& 6& \cellcolor{blue!25}10\\
\toprule
 19& \cellcolor{red!25}21& \cellcolor{red!25}25&29\\
 21& \cellcolor{red!25}23& \cellcolor{red!25}27&31\\
 1& \cellcolor{red!25}3& \cellcolor{red!25}7&11\\ 
 \cellcolor{green!25} 9&11&15&19\\
\bottomrule
\end{tabular}    
\end{center}
\caption{$(\alpha,\beta)$}\label{tab:forms 1}
\end{subtable}
\begin{subtable}{0.24\textwidth}
\begin{center}
\begin{tabular}{c|ccc}
  & 0& 4& \cellcolor{blue!25}8\\
\toprule
 0& \cellcolor{red!25}0& \cellcolor{red!25}4&8\\
 18& \cellcolor{red!25}18& \cellcolor{red!25}22&26\\
 20& \cellcolor{red!25}20& \cellcolor{red!25}24&28\\ 
 \cellcolor{green!25} 8&8&12&16\\
\bottomrule
\end{tabular}    
\end{center}
\caption{$\normal(\alpha,\beta)$}\label{tab:forms 2}
\end{subtable}
\begin{subtable}{0.24\textwidth}
\begin{center}
\begin{tabular}{c|ccc}
  & 10& 8& \cellcolor{blue!25}6\\
\toprule
 10& \cellcolor{red!25}20& \cellcolor{red!25}18&16\\
 1& \cellcolor{red!25}11& \cellcolor{red!25}9&7\\
 0& \cellcolor{red!25}10& \cellcolor{red!25}18&16\\ 
 \cellcolor{green!25} 6&16&14&12\\
\bottomrule
\end{tabular}    
\end{center}
\caption{$\negate \circ \normal (\alpha,\beta)$}\label{tab:forms 3}
\end{subtable}
\begin{subtable}{0.24\textwidth}
\begin{center}
\begin{tabular}{c|ccc}
  & 2& 4& \cellcolor{blue!25}0\\
\toprule
 0& \cellcolor{red!25}2& \cellcolor{red!25}4&0\\
 1& \cellcolor{red!25}3& \cellcolor{red!25}5&1\\
 10& \cellcolor{red!25}12& \cellcolor{red!25}14&10\\ 
 \cellcolor{green!25} 6&8&10&6\\
\bottomrule
\end{tabular}    
\end{center}
\caption{$\canonical(A)$}\label{tab:forms 4}
\end{subtable}
\caption{Different forms of the degree table $(\alpha,\beta)$.}
\label{tab:forms}
\end{table*}

\section{Bounds for the Degree Table} \label{sec: sumsets and lower bounds}

In this section, we present lower bounds for the number of distinct terms, $N$, in the degree table and upper bounds for its largest terms. The lower bounds are inherently interesting since, after all, our goal is to minimize $N$. The upper bounds on the entries are used in Section \ref{sec: linear programming} when formulating the degree table as an integer linear programming problem and allow us to find an optimal degree table via a finite search.

\subsection{Lower Bound on the Number of Distinct Terms} \label{sec: lower bounds n}

In this section, we prove the three inequalities in Theorem~\ref{theo:lb}. They are stated as Theorem~\ref{theo:degree tablelowerbound}, Corollary~\ref{cor:degree tablelowerbound_improvedbyone}, and Theorem~\ref{theo:degree tablelowerbound2_GASPoptT1}, respectively. The main technique behind the proof of these bounds comes from the theory of sumsets. In this context, the following lemma is well known (see~\cite{geroldinger2009combinatorial} and~\cite[Lemma~5.3, Proposition~5.8]{tao2006additive}). For completeness, we present a proof in the appendix. 

\begin{lemma}[\cite{geroldinger2009combinatorial,tao2006additive}]\label{lem:sumsetbound}
Let $A$ and $B$ be sets of integers.
Then $|A|+|B|-1 \le |A+B|$ and if $2 \le |A|, |B|$, then equality holds if and only if $A$ and $B$ are arithmetic progressions with the same common difference, i.e. $A=a+d[m]$ and $B=b+d[n]$ for $a,b,d,m,n \in \mathbb{Z}$.
\end{lemma}

\begin{proof}
See the Appendix.
\end{proof}

Using Lemma~\ref{lem:sumsetbound}, we give a lower bound on the sizes of degree tables.

\begin{theorem}\label{theo:degree tablelowerbound}
Let $(\alpha,\beta) \in \mathcal{A}(K,L,T)$. Then, it holds that $KL+\max\{K,L\}+2T-1 \le \N(\alpha,\beta)$.
\end{theorem}
\begin{proof}
Let $(\alpha,\beta) = (\ap,\as,\bp,\bs)$.
We obtain, by omitting $\SET(\as)+\SET(\bp)$ and by Property~(3) of Definition~\ref{def:degree table} that
\begin{align*}
\SET(\ap|\as)+\SET(\bp|\bs)
&=
(\SET(\ap)+\SET(\bp)) \dot\cup \Big((\SET(\as)+\SET(\bp))\cup (\SET(\ap|\as)+\SET(\bs))\Big)
\\
&
\supseteq
(\SET(\ap)+\SET(\bp)) \dot\cup (\SET(\ap|\as)+\SET(\bs))
\end{align*}
Again by Property~(3) of Definition~\ref{def:degree table}, we have $|\SET(\ap)+\SET(\bp)|=KL$ and $|\SET(\ap|\as)+\SET(\bs)| \ge (K+T)+(T)-1$ follows by Lemma~\ref{lem:sumsetbound}.
Hence, we get $\N(\ap,\as,\bp,\bs) \ge KL+K+2T-1$ and $\N(\ap,\as,\bp,\bs) \ge KL+L+2T-1$ follows by the same argumentation and omitting $\SET(\ap)+\SET(\bs)$.
This concludes the proof.
\end{proof}


The bound in Theorem \ref{theo:degree tablelowerbound} is tight if $\min\{K,L\}=1$ and is attained by $\mathsf{GASP}_{1}$. The application of Lemma~\ref{lem:sumsetbound} also allows us to get some insight into the structure of degree tables whenever the bound in Theorem~\ref{theo:degree tablelowerbound} is tight.

\begin{lemma}\label{lem:degree tablelowerbound_structure}
Let $K$, $L$, and $T$ be positive integers with $L \le K$ and $(\ap,\as,\bp,\bs)$ be a degree table in $\A(K,L,T)$ with $\N(\ap,\as,\bp,\bs) = KL+K+2T-1$.
Then there are nonnegative integers $a$, $b$, $d$, $x$, $y$, $e$ such that $1 \le d$, $n=T-1$, $m=K+T-1$ and
\begin{enumerate}
\item $\SET(\ap|\as) = \SET(\ap) \dot\cup \SET(\as) = a+d[m]$,
\item $\SET(\bs) = b+d[n]$,
\item $\SET(\as)+\SET(\bp) \subseteq \SET(\ap|\as)+\SET(\bs) = a+b+d[m+n]$,
\item $\SET(\bp) \subseteq b+d(\{-K,\ldots,-1\} \dot\cup \{n+1,\ldots,m+n\})$,
\item $\SET(\ap)+\SET(\bp) \subseteq \SET(\ap|\as)+\SET(\bp) = a+b+d\{-K,\ldots,2m+n\}$,
\item $|\SET(\ap)+\SET(\bp)| \le |a+b+d\{-K,\ldots,2m+n\}| = 2m+n+K+1$, and
\item $KL \le 3K+3T-2$, i.e. $L \le 3+\frac{3T-2}{K}$.
\end{enumerate}
If additionally $K=L$, then
\begin{enumerate}
\setcounter{enumi}{7}
\item $\SET(\as)=x+e[n]$ and $\SET(\beta)=y+e[m]$,
\item $\SET(\alpha)+\SET(\bs) = \SET(\as)+\SET(\beta) = a+b+d[m+n] = x+y+e[m+n]$, i.e. $d=e$,
\item $\SET(\alpha)+\SET(\beta) = a+y+d[2m]$, and
\item $K=L=1$.
\end{enumerate}
\end{lemma}

\begin{proof}
See the Appendix.
\end{proof}

Using the equivalence relation in Definition~\ref{def:equivalence}, we can set, without loss of generality, $a=0$, $d=1$, and $b=K$ in Lemma~\ref{lem:degree tablelowerbound_structure}. The Bullet points~(7.) and~(11.) allow us to strengthen the bound of Theorem~\ref{theo:degree tablelowerbound}.

\begin{corollary}\label{cor:degree tablelowerbound_improvedbyone}
Let $(\alpha,\beta) \in \mathcal{A}(K,L,T)$ be such that $3\max\{K,L\}+3T-2 < KL$ or $2 \le K = L$. Then, $KL+\max\{K,L\}+2T \le \N(\alpha,\beta)$.
\end{corollary}

\begin{proof}
This follows immediately by Theorem~\ref{theo:degree tablelowerbound} and Lemma~\ref{lem:degree tablelowerbound_structure}.
\end{proof}

For the lower bound in Theorem~\ref{theo:degree tablelowerbound2_GASPoptT1}, we need the following lemma.

\begin{lemma}\label{lem:row_col_intersection_size}
Let $(\ap,\as,\bp,\bs) \in \A(K,L,T)$. Then, for all $1 \le i \le K+T$ and all $1 \le j \le L+T$, it holds that $|(\alpha_i+\SET(\beta_p)) \cap (\SET(\alpha_p)+\beta_j)| \le 1$.
\end{lemma}

\begin{proof}
See the Appendix.
\end{proof}

Lemma~\ref{lem:row_col_intersection_size} provides a new lower bound which shows the optimality of $\mathsf{GASP}_r$ for $T=1$.

\begin{theorem}\label{theo:degree tablelowerbound2_GASPoptT1}
Let $(\alpha,\beta) \in \A(K,L,T)$. Then, $KL+K+L+2T-1-T\min\{K,L,T\} \le \N(\alpha,\beta)$. Thus, for $\ap=(0,1,\ldots,K-1)$, $\as=(KL)$, $\bp=(0,K,\ldots,K(L-1))$, and $\bs=(KL)$, the degree table $(\ap,\as,\bp,\bs) \in \A(K,L,1)$ is of minimum size.
\end{theorem}

\begin{proof}
See the Appendix.
\end{proof}

The bound in Theorem~\ref{theo:degree tablelowerbound2_GASPoptT1} is stronger than the bound in Theorem~\ref{theo:degree tablelowerbound} if and only if
\begin{align*}
KL+\max\{K,L\}+2T-1 < KL+K+L+2T-1-T\min\{K,L,T\}.
\end{align*}
This is equivalent to $T^2 < \min\{K,L\}$.

\subsection{A Quadratic Upper Bound on the Entries of the Degree Table for Large $T$} \label{sec: linear bounds for entries}

In this section, we present an upper bound on the entries of the degree table for large values of the security parameter $T$. This upper bound is quadratic in the partitioning parameters, $K$ and $L$, and the security parameter, $T$. Upper bounding the entries in the degree table allows us to transform the problem into a finite search, as we do in Section \ref{sec: linear programming}. Our main tool in this section is the following corollary from \cite{MR1379398}, which we cite as a theorem due to its importance. For a finite set of integers $S$ we use the abbreviation $\ell(S)=\max S-\min S$.

\begin{theorem}[Corollary~2 in~\cite{MR1379398}]\label{theo:Stanchescu}
Let $A=\{a_1,\ldots,a_k\}$ and $B=\{b_1,\ldots,b_l\}$ be two sets of integers and $\delta= \mathbf{1}_{\ell(A)=\ell(B)}$.
Denote by $d$ the greatest common divisor of $a_2-a_1,\ldots,a_k-a_1,b_2-b_1,\ldots,b_l-b_1$ and put $a=\ell(A)/d$, $b=\ell(B)/d$.
If $N=|A+B|<|A|+|B|+\min\{|A|,|B|\}-2-\delta$, then $a\le N-l$ and $b \le N-k$.
\end{theorem}

This theorem has consequences for normal degree tables that allow us to bound the largest integer in the degree table, provided that a suitable upper bound on $\N(\alpha,\beta)$ is known, e.g. via Theorem~\ref{theo:GASP_r_N}. Thus, the determination of $\N(\alpha,\beta)$ and its associated binary linear program, which we present in Theorem~\ref{theo:BLP}, are finite problems.

\begin{lemma}\label{lem:largeT}
Let $(\alpha,\beta) \in \A(K,L,T)$ be a normal degree table and $\delta= \mathbf{1}_{\max\{\SET(\alpha)\}=\max\{\SET(\beta)\}}$.
If $\N(\alpha,\beta) \le K+L+\min\{K,L\}+3T-3-\delta$, then $\max\{\SET(\alpha)\} \le \N(\alpha,\beta)-L-T$ and $\max\{\SET(\beta)\} \le \N(\alpha,\beta)-K-T$.
\end{lemma}
\begin{proof}
Since $(\alpha,\beta)$ is normal, $\min\{\SET(\alpha)\}=0$, $\ell(\SET(\alpha))=\max\{\SET(\alpha)\}$, $\min\{\SET(\beta)\}=0$, $\ell(\SET(\beta))=\max\{\SET(\beta)\}$, the greatest common divisor is one, it follows that $a=\max\{\SET(\alpha)\}$, and $b=\max\{\SET(\beta)\}$.
Then the definition of $\N(\alpha,\beta)$ and application of Theorem~\ref{theo:Stanchescu} complete the proof.
\end{proof}

We now use Theorem~\ref{theo:GASP_r_N} to upper bound the largest elements in $\alpha$ and $\beta$.

\begin{theorem}\label{theo_finiteness_by_sumsets}
Let $K,L,T$ be such that
\begin{align} \label{math_inline_stanchescu}
2KL-K-L-\min\{K,L\}+3 \le T.
\end{align}
Then, any normal degree table $(\alpha^*,\beta^*) \in \A(K,L,T)$ with $\N(\alpha^*,\beta^*) = \N(K,L,T)$ has the property that the largest entry in $\alpha^*$ is at most $2KL+T-1-L$ and the largest entry in $\beta^*$ is at most $2KL+T-1-K$.
\end{theorem}
\begin{proof}
$\mathsf{GASP}_{\text{big}}$ (which is $\mathsf{GASP}_{r}$ with $r=\min\{K,T\}$) constructs a degree table $(\alpha,\beta)$ of size $2KL+2T-1$, cf. Theorem~\ref{theo:GASP_r_N} or~\cite[Theorem~3]{MR1379398}, and hence we have $\N(K,L,T) \le 2KL+2T-1$.
Next, \eqref{math_inline_stanchescu} provides a sufficient criterion for the application of Lemma~\ref{lem:largeT}, which completes the proof.
\end{proof}

We want to remark that it is possible to improve all three inequalities in the statement of Theorem~\ref{theo_finiteness_by_sumsets} if better upper bounds on $N$ are used.

Theorem~\ref{theo_finiteness_by_sumsets} allows us to perform an exhaustive search of all normal degree tables $(\alpha^*,\beta^*)$ meeting \eqref{math_inline_stanchescu} with $\N(\alpha^*,\beta^*) = \N(K,L,T)$. For example, consider the case where $K=L=2$ and $T=5$. Then, Theorem~\ref{theo_finiteness_by_sumsets} implies that the largest entry of $\alpha$ and $\beta$ is $M=10$. Hence, $\alpha,\beta \subseteq \{0,\ldots,10\}$ and there are $\binom{(M+1)-1}{K+T-1} \cdot \binom{K+T}{T} = 4410$ possible choices for $\alpha$ (similar for $\beta$) such that $\ap,\as,\bp,\bs$ are sorted and $0 \in \alpha$ as well as $0 \in \beta$.

Combining all $4410^2$ possibilities for $(\alpha,\beta)$ and filtering the cases in which Definition~\ref{def:degree table} (3) is violated and such that $\operatorname{GCD}(\SET(\alpha\mid\beta))=1$, we find that there are 2716 normal degree tables with $K=L=2$ and $T=5$ and of those, only 4 attain the minimum value of $N=17$, the same value achieved by $\mathsf{GASP}_{r^*} = \mathsf{GASP}_{2} = \mathsf{GASP}_{\text{big}}$. Moreover, from the bounds in Theorem~\ref{theo:lb} we obtain $15 \le N$, $16 \le N$, and $7 \le N$. 

Thus, we just showed that for $K=L=2$ and $T=5$, \GASP{r} is optimal, and that the bounds in Theorem~\ref{theo:lb} are not tight in general. The four optimal degree tables are presented in Table \ref{tab: optimal}. 

\begin{table*}[!t]
\setlength{\tabcolsep}{0.33em}
\begin{subtable}{0.24\textwidth}
\resizebox{3.5cm}{!}{\begin{tabular}{c|ccccccc}
   &  7&  8&  \cellcolor{blue!25}0& \cellcolor{blue!25}1&\cellcolor{blue!25}2& \cellcolor{blue!25}3& \cellcolor{blue!25}4\\\toprule
 6&  \cellcolor{red!25}13&  \cellcolor{red!25}14&  6& 7&8&9&10\\
 8&  \cellcolor{red!25}15& \cellcolor{red!25}16&  8& 9&10&11&12\\
 \cellcolor{green!25}0&  7&  8& 0& 1&2&3&4\\
 \cellcolor{green!25}1&  8&  9& 1& 2&3&4&5\\
 \cellcolor{green!25}2& 9& 10&2&3&4&5&6\\
 \cellcolor{green!25}3& 10& 11& 3& 4&5&6&7\\
\cellcolor{green!25}4& 11& 12& 4& 5&6&7&8\\
\bottomrule
\end{tabular}}
\caption{}\label{tab:optimal 1}
\end{subtable}
\begin{subtable}{0.24\textwidth}
\resizebox{3.5cm}{!}{\begin{tabular}{c|ccccccc}
    &  6&  8&  \cellcolor{blue!25}0& \cellcolor{blue!25}1&\cellcolor{blue!25}2& \cellcolor{blue!25}3& \cellcolor{blue!25}4\\\toprule
 7&  \cellcolor{red!25}13& \cellcolor{red!25}15& 7&8&9&10&11\\
 8&  \cellcolor{red!25}14& \cellcolor{red!25}16&  8& 9&10&11&12\\
 \cellcolor{green!25}0&  6&  8& 0& 1&2&3&4\\
 \cellcolor{green!25}1&  7&  9& 1& 2&3&4&5\\
 \cellcolor{green!25}2& 8& 10&2&3&4&5&6\\
 \cellcolor{green!25}3& 9& 11& 3& 4&5&6&7\\
 \cellcolor{green!25}4& 10& 12& 4& 5&6&7&8\\
\bottomrule
\end{tabular}}
\caption{}\label{tab:optimal 2}
\end{subtable}
\begin{subtable}{0.24\textwidth}
\resizebox{3.65cm}{!}{\begin{tabular}{c|ccccccc}
   &  0&  2&  \cellcolor{blue!25}4& \cellcolor{blue!25}5& \cellcolor{blue!25}6& \cellcolor{blue!25}7& \cellcolor{blue!25}8\\\toprule
 0&  \cellcolor{red!25}0&  \cellcolor{red!25}2& 4&5&6&7&8\\
 1&  \cellcolor{red!25}1& \cellcolor{red!25}3&  5&6&7&8&9\\
 \cellcolor{green!25}4&  4&  6& 8& 9&10&11&12\\
 \cellcolor{green!25}5&  5&  7& 9& 10&11&12&13\\
 \cellcolor{green!25}6& 6& 8&10&11&12&13&14\\
 \cellcolor{green!25}7& 7& 9& 11& 12&13&14&15\\
\cellcolor{green!25}8& 8& 10& 12& 13&14&15&16\\
\bottomrule
\end{tabular}}
\caption{}\label{tab:optimal 3}
\end{subtable}
\begin{subtable}{0.24\textwidth}
\resizebox{3.5cm}{!}{\begin{tabular}{c|ccccccc}
   &  0&  1&  \cellcolor{blue!25}4& \cellcolor{blue!25}5& \cellcolor{blue!25}6& \cellcolor{blue!25}7& \cellcolor{blue!25}8\\\toprule
 0&  \cellcolor{red!25}0&  \cellcolor{red!25}1& 4&5&6&7&8\\
 2&  \cellcolor{red!25}2& \cellcolor{red!25}3&  6&7&8&9&10\\
 \cellcolor{green!25}4&  4&  5& 8& 9&10&11&12\\
 \cellcolor{green!25}5&  5&  6& 9& 10&11&12&13\\
 \cellcolor{green!25}6& 6& 7&10&11&12&13&14\\
 \cellcolor{green!25}7& 7& 8& 11& 12&13&14&15\\
\cellcolor{green!25}8& 8& 9& 12& 13&14&15&16\\
\bottomrule
\end{tabular}}
\caption{}\label{tab:optimal 4}
\end{subtable}
\caption{The four optimal normal degree tables for $K=L=2$ and $T=5$, all achieving $N=17$. Note that (c) is $\mathsf{GASP}_2$ and is such that $(c) = \normal \circ \negate (b)$. Also, (d) is the transpose of (c), and such that $(d)  = \normal \circ \negate (a)$.}
\label{tab: optimal}
\end{table*}

\subsection{An Operational Upper Bound on the Entries of the Degree Table} \label{sec: operational bounds for entries}

In order for the degree table to meaningfully solve an SDMM problem, it cannot take longer to compute with the numbers in the degree table than it would to perform the matrix multiplication locally. This reasoning provides us with an \say{operational} upper bound for the entries in the degree table, i.e. if the upper bound is not satisfied, it is better to perform the computation locally.

Here we count operations, i.e. additions and multiplications, over $\mathbb{F}_q$. As $\mathsf{GASP}_{r}$ extends the field, we compare the case that a single field operation over the required extension field is more expensive than performing the whole matrix multiplication locally.

\begin{theorem}\label{theo_finite_involving_q}
Let $(\alpha,\beta) \in \mathcal{A}$ be a degree table used to perform SDMM on two matrices $A \in \mathbb{F}_q^{a \times b}$ and $B \in \mathbb{F}_q^{b \times c}$. If any entry in $\alpha$ or $\beta$ is at least $q^{2abc-ac}-2$, then, performing the multiplication locally is more efficient than performing SDMM with $(\alpha,\beta)$.
\end{theorem}

\begin{proof}
The trivial matrix multiplication algorithm, i.e. $C_{ij} = \sum_{k=1}^{s} A_{ik}B_{kj}$, uses $b-1$ additions and $b$ multiplications for the computation of each of the $ac$ elements, i.e. a total of $2abc-ac$ operations. Let $M$ be the largest entry in $\alpha$ or $\beta$. The size of the extension field $q'$ required by $\mathsf{GASP}_{r}$ is at least $M+2$ as $x^{q'-1} = 1$ for all $x \in \mathbb{F}_{q'}^{\times}$ implies that all entries in $\alpha$ and $\beta$ are in the ring $\mathbb{Z} / ((q'-1)\mathbb{Z})$ and the canonical set of representatives $\{0,\ldots,q'-2\}$ can be chosen.

A single operation in $\mathbb{F}_{q'}$ with $q'=q^e$ takes at least $e=\log_q(q') \ge \log_q(M+2)$ operations in $\mathbb{F}_q$. Hence, the number of operations performing the matrix multiplication locally using the trivial algorithm is at most the number of operations of the application of $\mathsf{GASP}_{r}$ if $2abc-ac \le \log_q(M+2)$ which is equivalent to $q^{2abc-ac}-2 \le M$.

\end{proof}

\section{Degree Tables via Integer Linear Programming} \label{sec: linear programming}

In this section we formulate the problem of finding a degree table with a low number of distinct entries as a binary linear program (BLP) problem. We first look at the general case in which the BLP has an infinite number of variables and constraints. The problem can be made finite by bounding the entries in the degree table, as done in Sections \ref{sec: linear bounds for entries} and \ref{sec: operational bounds for entries}. We then consider a special case where $\ap$ an $\beta$ are fixed. This allows us to run a greedy algorithm which is able to find degree tables which are different from \GASP{r} but attain the same number of distinct terms. We have not found any degree table which outperforms \GASP{r}.

\subsection{An Infinite Binary Linear Program for Finding Degree Tables}

In this section we present a binary linear program with an infinite number of variables and constraints which can compute $\N(K,L,T)$. The number of variables and constraints can be made finite by bounding the values of the entries in the degree table, e.g. as shown in Sections~\ref{sec: linear bounds for entries} and~\ref{sec: operational bounds for entries}.


\begin{theorem}\label{theo:BLP}
The minimum number of distinct terms in a degree table, $\N(K,L,T)$, is the optimum value of the following BLP, where $\lambda = \min\{K,L\}+T$.
\begin{align}
\text{Minimize}\quad &\sum_{e} U_e  \label{BLP_obj}\\
\text{subject to:} \hspace{5em} M_{r,c,e} &\le U_e && \forall r,c,e \label{BLP_lbU}\\
\sum_{(r,c) \ne (r',c')} M_{r,c,e} &\le \lambda \cdot (1-M_{r',c',e}) && \forall e, r' \le K, c' \le L \label{BLP_decodability}\\
\sum_{r} R_{r,e} &\le 1 && \forall e \label{BLP_distinct_rows}\\
\sum_{c} C_{c,e} &\le 1 && \forall e \label{BLP_distinct_columns}\\
\sum_{e} M_{r,c,e} &= 1 && \forall r,c \label{BLP_M1}\\
\sum_{e} R_{r,e} &= 1 && \forall r \label{BLP_R1}\\
\sum_{e} C_{c,e} &= 1 && \forall c \label{BLP_C1}\\
\sum_{e} e \cdot M_{r,c,e} &= \sum_{e} e \cdot R_{r,e} + \sum_{e} e \cdot C_{c,e} && \forall r,c \label{BLP_computeM}\\
U_{e},R_{r,e},&C_{c,e},M_{r,c,e} \in \{0,1\} && \forall r,c,e \label{BLP_var}
\end{align}
\end{theorem}

\begin{proof}
The variables $R_{r,e}$ and $C_{c,e}$ are one iff $\alpha_r=e$ respectively $\beta_c=e$.
To ensure that exactly one $e$ value is assigned to the same $\alpha_r$ or $\beta_c$, we have the Equations~(\ref{BLP_R1}) and~(\ref{BLP_C1}).
In this setting, the Inequalities~(\ref{BLP_distinct_rows}) and~(\ref{BLP_distinct_columns}) ensure Property~(1) and~(2) of Definition~\ref{def:degree table}.

The variables $M_{r,c,e}$ are one iff $\alpha_r+\beta_c=e$ and this sum is represented in Equation~(\ref{BLP_computeM}).
Equation~(\ref{BLP_M1}) implies then that for each $r,c$ exactly one $M_{r,c,e}$ is one.
Property~(3) of Definition~\ref{def:degree table} is handled by the Inequality~(\ref{BLP_decodability}):
If $M_{r',c',e}$ is one, i.e. $\alpha_{r'}+\beta_{c'}=e$, then there may not be $(r,c) \ne (r',c')$ with $\alpha_{r}+\beta_{c}=e$, i.e. $M_{r,c,e}=1$, so in that case the right hand side is zero, implying that all variables on the left hand side must be zero.
Conversely, if $M_{r',c',e}$ is zero, then the right hand side is $\min\{K,L\}+T$ which imposes no restriction on the left hand side.

The variables $U_e$ are one (in any optimal solution) iff there is an $(r,c)$ with $\alpha_{r}+\beta_{c}=e$.
This logic is ensured since Inequality~(\ref{BLP_lbU}) implies that $U_e$ is one if there is an $(r,c)$ with $\alpha_{r}+\beta_{c}=e$ and the converse is provided by the objective function, which in turn is then exactly the size of the degree table.
\end{proof}

The invariance of the sizes of degree tables under the equivalence in Definition~\ref{def:equivalence}, allows us to incorporate additional constraints in the BLP of Theorem~\ref{theo:BLP} in order to reduce the search space and to mitigate the impact of symmetries.
\begin{align}
\sum_e e \cdot R_{r,e}+1 &\le \sum_e e \cdot R_{r+1,e} && \forall r < K \label{BLP_wlog_sort_ap}\\
\sum_e e \cdot R_{r,e}+1 &\le \sum_e e \cdot R_{r+1,e} && \forall K < r < K+T \label{BLP_wlog_sort_as}\\
\sum_e e \cdot C_{c,e}+1 &\le \sum_e e \cdot C_{c+1,e} && \forall c < L \label{BLP_wlog_sort_bp}\\
\sum_e e \cdot C_{c,e}+1 &\le \sum_e e \cdot C_{c+1,e} && \forall L < c < L+T \label{BLP_wlog_sort_bs}\\
R_{1,0}+R_{K+1,0} &= 1 \label{BLP_wlog_a0}\\
C_{1,0}+C_{L+1,0} &= 1 \label{BLP_wlog_b0}
\end{align}
Inequalities~\ref{BLP_wlog_sort_ap},~\ref{BLP_wlog_sort_as},~\ref{BLP_wlog_sort_bp}, and~\ref{BLP_wlog_sort_bs} ensure that $\ap$, $\as$, $\bp$, and $\bs$ are sorted.
Equations~\ref{BLP_wlog_a0} and~\ref{BLP_wlog_b0} force $\min\{\SET(\alpha)\}=0$ and $\min\{\SET(\beta)\}=0$.

\subsection{The special case of fixed \texorpdfstring{$\ap$}{alpha prefix} and \texorpdfstring{$\beta$}{beta}}\label{sec:fixed}

In this section, we fix $\ap=(0,1,\ldots,K-1)$, $\bp=(0,K,\ldots,K(L-1))$, and $\bs=(KL,KL+1,\ldots,KL+T-1)$. Note that $\mathsf{GASP}_{r}$ uses this structure and that $\SET(\ap)+\SET(\bp)=[KL-1]$ and $\SET(\ap)+\SET(\bp|\bs)=[KL+K+T-2]$. From Corollary~\ref{cor:finite} we immediately obtain the bound $\max\{\SET(\alpha)\} \le (K+T-1)(KL+T)$, which can be improved since we fixed $\ap$.

\begin{lemma}\label{lem:finite2}
Let $(\ap,\as,\bp,\bs) \in \A(K,L,T)$ be a normal degree table such that $\ap=(0,1,\ldots,K-1)$, $\bp=(0,K,\ldots,K(L-1))$, and $\bs=(KL,KL+1,\ldots,KL+T-1)$. Then, it holds that $\SET(\as) \subseteq \{KL,KL+1,\ldots,T(KL+T) + K-1\}$.
\end{lemma}
\begin{proof}
Since $\SET(\ap)+\SET(\bp)=[KL-1]$, we immediately have $\SET(\as) \subseteq \{KL,KL+1,\ldots\}$ for all $\as$.
Since $(\ap,\as,\bp,\bs)$ is a normal degree table, we have that $(\ap,\as,\bp,\bs)$ is a squeezed degree table and that $\as$ is sorted in increasing order.
Then Corollary~\ref{cor:finite} shows that $\pi\circ\alpha_{i+1}-1\le\pi\circ\alpha_i+KL+T-1$ for all integers $i$ with $1 \le i < K+T$ and a sorting permutation $\pi$ with respect to $\alpha$.
Note, that since $\ap|\as$ is sorted, $\pi$ is the identity.
Hence, we have $\alpha_{i+1}-\alpha_i \le KL+T$ and 
\begin{align*}
&\alpha_{K+T}-\alpha_{K}
=
\sum_{i=K}^{K+T-1} (\alpha_{i+1}-\alpha_i)
\\&\le
\sum_{i=K}^{K+T-1} (KL+T)
=
T(KL+T),
\end{align*}
i.e. $\alpha_{K+T} \le T(KL+T) + K-1$.

\end{proof}

Lemma~\ref{lem:finite2} allows the application of the BLP of Theorem~\ref{theo:BLP} to get the best upper bound on $\N(K,L,T)$ for fixed $\ap$ and $\beta$, but it also enables us to give a greedy algorithm for finding an upper bound on $\N(K,L,T)$.

Theorem~\ref{theo:BLP} can be simplified and improved in the setting of fixed $\ap$ and $\beta$.

\begin{corollary}\label{cor:BLP_special_case}
Let $L \le K$, $\ap=(0,1,\ldots,K-1)$, $\bp=(0,K,\ldots,K(L-1))$, and $\bs=(KL,KL+1,\ldots,KL+T-1)$.
Then, the minimum size of a degree table $(\ap,\as,\bp,\bs) \in \A(K,L,T)$ is exactly the optimum value of the following integer linear program (ILP).
We abbreviate
$\mathcal{R}=\{K+1,\ldots,K+T\}$,
$\mathcal{C}=\{1,\ldots,L+T\}$,
$\mathcal{F}=\{KL,\ldots,KL+K+T-2\}$,
$\mathcal{E}=\{KL,\ldots,(T+1)(KL+T)+K-2\}$, and
$\mathcal{V}=\{KL,\ldots,T(KL+T)+K-1\}$.
\begin{align}
\min{} &N \text{ st} \label{BLP_special_o} \\
N &= KL+\sum_{e \in \mathcal{E}} U_e \label{BLP_special_computeN} \\
U_e &= 1 && \forall e \in \mathcal{F} \label{BLP_special_fixU} \\
(L+T)S_{r,s} &\le \sum_{c \in \mathcal{C}} U_{s+\beta_{c}} && \forall r \in \mathcal{R}, s \in \mathcal{V} \label{BLP_special_SimplyU} \\
\sum_{s \in \mathcal{V}} S_{r,s} &= 1 && \forall r \in \mathcal{R} \label{BLP_special_S1} \\
\sum_{s \in \mathcal{V}} sS_{r,s} &= R_{r} && \forall r \in \mathcal{R} \label{BLP_special_SR} \\
R_{i}+1 &\le R_{i+1} \le R_{i}+KL+T && \forall i \in \mathcal{R}\setminus \{K+T\} \label{BLP_special_wlog} \\
N &\in \mathbb{Z} \label{BLP_special_N} \\
U_{e} &\in \{0,1\} && \forall e\in\mathcal{E} \label{BLP_special_U} \\
R_{r} &\in \mathcal{V} && \forall r\in\mathcal{R} \label{BLP_special_R} \\
S_{r,s} &\in \{0,1\} && \forall r \in \mathcal{R}, s \in \mathcal{V} \label{BLP_special_S}
\end{align}
\end{corollary}
\begin{proof}
The interpretation of the variables is:
\begin{align*}
&
S_{r,s} = 1 \Leftrightarrow R_{r} = s \Leftrightarrow (\as)_r = s,
\\&
U_{e} = 1 \Leftrightarrow e \in (\SET(\ap)+\SET(\bs)) \cup (\SET(\as)+\SET(\bp|\bs)),
\\&
N = \N(\ap,\as,\bp,\bs).
\end{align*}
Next,
$\mathcal{R}$ is the set of indices of $\as$ in $\ap|\as$,
$\mathcal{C}$ is the set of indices of $\bp|\bs$,
$\mathcal{F}=\SET(\ap)+\SET(\bs)$ is the set of entries in the top right corner of the degree table and independent of $\as$,
$\mathcal{V}$ is the set of entries for $\as$ implied by Lemma~\ref{lem:finite2}, and
$\mathcal{E}=\{\min\{\mathcal{V}\}+\min\{\SET(\bp|\bs)\}, \ldots, \max\{\mathcal{V}\}+\max\{\SET(\bp|\bs)\}\}$ the set of possible entries in the degree table excluding the top left corner.

Now, the remainder of the proof follows by rewriting the inequalities of Theorem~\ref{theo:BLP}.


We call $R_{r,e}$ in Theorem~\ref{theo:BLP} here $S_{r,e}$ and use $R_r = \sum_e e \cdot S_{r,e}$ so that Inequality~(\ref{BLP_wlog_sort_as}) becomes the left part of Inequality~(\ref{BLP_special_wlog}).
The right part of Inequality~(\ref{BLP_special_wlog}) is given by Corollary~\ref{cor:finite} since we may assume wlog. a normal degree table.

Since $\SET(\ap)+\SET(\bp)=[KL-1]$ and $\SET(\as)+\SET(\bp|\bs)\subseteq \{KL,\ldots,T(KL+T+1)+K+KL-2\}$ we get Inequality~(\ref{BLP_decodability}) of the BLP in Theorem~\ref{theo:BLP}.

Then, $\mathcal{F}$ and the offset in the objective function follow since $\SET(\ap)+\SET(\bp|\bs)=[KL+K+T-2]$ and $\min\{\mathcal{V}+\SET(\bp|\bs)\}=KL$, so we omit $U_e=1$ for $e < KL$.

Inequality~(\ref{BLP_special_SimplyU}) is built by using the Inequalities~(\ref{BLP_lbU}) and~(\ref{BLP_computeM}).
If the left hand side is zero, i.e. $S_{r,s}=0$, then the right hand side is not restricted.
Else the left hand side is $L+T$ and, using $|\mathcal{C}|=L+T$, we have $U_{s+b_{c}}=1$ for all $c \in \mathcal{C}$, i.e. $s+\SET(\beta) \subseteq \SET(\alpha)+\SET(\beta)$.\footnote{Of course $S_{r,s} \le U_{s+b_{c}}$ for all $r \in \mathcal{R}$, $s \in \mathcal{V}$, and $c \in \mathcal{C}$ can be used instead of Inequality~(\ref{BLP_special_SimplyU}), which is even a better formulation. But here, we favor a smaller number of constrains.}
\end{proof}

Solving this BLP for each $1 \le L \le K \le 9$, $1 \le T \le 9$ using \texttt{CPLEX} takes at most two weeks on a high-end computer. The objective value equals the size of a degree table constructed by \GASP{r} with $r=r^*$. Note that this BLP has $T^2KL + T^3 + T^2 + TK + 2T + K$ variables and $T^2KL + T^3 - TKL - TK + 5T + K - 3$ constraints, which makes it difficult to solve exactly by sheer size, but also other aspects might render the BLP difficult, such as poor LP-relaxations in the subproblems arising in an underlying branch and bound algorithm or symmetry.

\subsection{Greedy algorithm}

In this section, we present a greedy algorithm which can find good degree tables for large parameters. As in the previous section, we assume $(\ap,\as,\bp,\bs) \in \A(K,L,T)$ satisfies $\ap=(0,1,\ldots,K-1)$, $\bp=(0,K,\ldots,K(L-1))$, and $\bs=(KL,KL+1,\ldots,KL+T-1)$.

The greedy algorithm in the version depicted here computes and maintains an integral set $S_i$ for $i \in \mathcal{V} = \{KL,KL+1,\ldots,T(KL+T) + K-1\}$ (cf. Lemma~\ref{lem:finite2}) and traverses a search tree in depth-first search order in which only the set $\operatorname{argmax}\{|S_i| \mid i \subseteq \mathcal{V}\}$ is taken into consideration for branching.
The maximum depth of the tree is $T$.
As soon as a solution $\as$ is found, the size of the corresponding degree table is used for pruning.

\begin{algorithm}
\caption{Greedy search strategy to find $\as$ such that the degree table has small size.}
\label{alg:greedy}
\begin{algorithmic}[1]
\algnewcommand{\LineComment}[1]{\State \(\triangleright\) #1}
\Procedure{Greedy}{$K,L,T$}
\State initialize $\ap, \bp, \bs, N_{\text{best}}=\infty, \mathcal{V}$ as described above
\For{$i \in \mathcal{V}$}
\LineComment{for each possible entry in $\as$, we compute the set of integers which are already contained in the top part of the degree table}
\State $S_i \gets (i+\beta)\cap(\ap+\beta)$
\EndFor
\State \Call{Greedy\_Recursion}{$(),\{S_i \mid i \in \mathcal{V}\}$}
\EndProcedure
\Procedure{Greedy\_Recursion}{$\as,\{S_i \mid i \in \mathcal{V}\}$}
\If{$\text{size of the partial degree table}+T-|\as|>N_{\text{best}}$} \Comment{pruning}
\State \Return
\EndIf
\If{$|\as| = T$} \Comment{reached a leaf}
\If{$\text{size of the degree table}<N_{\text{best}}$}
\State update $N_{\text{best}}$ and save $\as$
\EndIf
\State \Return
\EndIf
\LineComment{we try all possibilities for the next element in $\as$ among all not yet chosen elements such that the number of integers which are not contained in the fixed part of the degree table is minimal}
\For{$r \in \operatorname{argmax}\{|S_i| \mid i \in \mathcal{V} \setminus \SET(\as) \}$}
\For{$i \in \mathcal{V}$}
\LineComment{update the set of integers which are contained in the fixed part of the degree table if $r$ is included in $\SET(\as)$}
\State $T_i \gets S_i \cup ((i+\beta)\cap(r+\beta))$
\EndFor
\State \Call{Greedy\_Recursion}{$\as|(r),\{T_i \mid i \in \mathcal{V}\}$}
\EndFor
\EndProcedure
\end{algorithmic}
\end{algorithm}

This algorithm can be slightly modified to find all $\as$ which minimize the sizes of the corresponding degree tables. Furthermore, it allows to find small degree tables such that $\as$ is not a generalized arithmetic progression, as is the case for \GASP{r}.
\begin{example}
Consider $K=L=T=15$. Then, \GASP{r^*} constructs a degree table where the number of distinct entries is $N=368$. The following degree table with the same amount of distinct entries can be found by applying the greedy algorithm.
\begin{align*}
\ap &= (0,1,2,3,4,5,6,7,8,9,10,11,12,13,14), \\
\as &= (225,226,227,229,240,241,242,244,255,256,257,259,270,271,272),\\
\bp &= (0,15,30,45,60,75,90,105,120,135,150,165,180,195,210),\\
\bs &= (225,226,227,228,229,230,231,232,233,234,235,236,237,238,239).\\
\end{align*}
\end{example}

In our computational search for \say{good} degree tables we have found many degree tables with the same performance of \GASP{r^*}. However, we have found none that perform better.

\section{On the Inner vs. Outer Product Partitioning} \label{sec: inner vs outer}

As mentioned in the introduction, minimizing the communication costs is equivalent to minimizing the minimum amount of servers $N$ when a fixed matrix partitioning is set. This is not true, however, when comparing SDMM schemes which partition the matrix differently. And thus, for different partitionings, the expressions for the communication costs, and not just the amount of servers $N$, must be compared.

In this paper we focused on the outer product partitioning (OPP) given by \eqref{partition1}. Another partitioning often considered in the literature (e.g., \cite{mital2020secure}) is the inner product partitioning (IPP) given by
\begin{align} \label{eq: inner}
\begin{aligned}
A = \begin{bmatrix}A_1 & \cdots & A_M\end{bmatrix} \quad \text{and} \quad
B = \begin{bmatrix}B_1 \\ \vdots \\ B_M\end{bmatrix}, \quad
\text{so that} \quad
AB = \sum_{i=1}^M A_i B_i .
\end{aligned}
\end{align}

In this section we compare the communication costs, including both upload and download, and show that OPP outperforms IPP when the number $b$ of columns of $A$ (or rows of $B$) is not too big. In particular, we show that, for square matrices, OPP has a better asymptotic total communication cost than IPP.

Our analysis assumes that the OPP scheme is such that it obtains a recovery threshold of $N_o = \Theta(KL)$ and that the IPP scheme is such that it obtains a recovery threshold of $N_i = \Theta(M)$, where, in order to simplify our presentation, we assume that the number of colluding servers, $T$, is constant. Also, in order to make a fair comparison, we assume that each server performs the same amount of computations in both schemes, i.e. $KL=M$. Under these assumptions, we obtain the following communication costs.

We first recall that $A$ has dimensions $a \times b$ and $B$ has dimensions $b \times c$. Thus, in the OPP setting, the matrices $A_i$ and $B_j$ in \eqref{partition1} have dimensions $\frac{a}{K} \times b$ and $b \times \frac{c}{L}$, respectively. The user must send one of these matrices to each server, and thus, the total upload cost is given by $U_O = N_O \left( \frac{ab}{K} + \frac{bc}{L} \right)$. Under the OPP the user must download one matrix of the form $A_i B_j$ from each server, giving a total download of $D_O = N_O \frac{ac}{KL}$. Using similar arguments for the IPP we obtain an upload cost of $U_I = N_I \left( \frac{ab}{M} + \frac{bc}{M} \right)$ and a download cost of $D_I = N_I ac$.

We set a variable $n$ to control the growth of all the other parameters. We set nonnegative constants $\varepsilon_a, \varepsilon_b, \varepsilon_c, \varepsilon_K, \varepsilon_L, \varepsilon_M \in \mathbb{R}$ and set $a=n^{\varepsilon_a}, b=n^{\varepsilon_b}, c=n^{\varepsilon_c}, K=n^{\varepsilon_K}, L=n^{\varepsilon_L}$, and $M=n^{\varepsilon_M}$. So, for example, if we want to analyze the case where $A$ and $B$ are square matrices, we need only to take $\varepsilon_a=\varepsilon_b=\varepsilon_c=1$. Note that, since $K$, $L$, and $M$ are partitioning parameters, they cannot exceed their corresponding matrix dimension, i.e. $\varepsilon_K \leq \varepsilon_a$, $\varepsilon_L \leq \varepsilon_c$, and $\varepsilon_M \leq \varepsilon_b$. Also, since $KL=M$, it holds that $\varepsilon_K + \varepsilon_L = \varepsilon_M$.

Under these parameters we then obtain $U_O = \Theta (n^{\max \{ \varepsilon_a + \varepsilon_b +\varepsilon_L , \varepsilon_b + \varepsilon_c +\varepsilon_K \} })$, $D_O = \Theta(n^{\varepsilon_a + \varepsilon_c})$, $U_I = \Theta (n^{\max \{ \varepsilon_a + \varepsilon_b, \varepsilon_b + \varepsilon_c \} })$, and $D_I = \Theta(n^{\varepsilon_a + \varepsilon_c + \varepsilon_M})$. To find for which parameters OPP outperforms IPP we must solve $\Theta(U_O + D_O) \leq \Theta(U_I + D_I)$. A direct computation shows that the inequality holds if and only if $\varepsilon_b \leq \min \{ \varepsilon_a + \varepsilon_L , \varepsilon_c + \varepsilon_K \}$. Thus, intuitively, OPP outperforms IPP when the number $b$ of columns of $A$ (or rows of $B$) is not too big. 

As a simple example, if we set $\varepsilon_a=\varepsilon_b=\varepsilon_c=1$, and $\varepsilon_K = \varepsilon_L = \frac{\varepsilon}{2}$, then $U_O + D_O = \Theta(n^{2 + \frac{\varepsilon}{2}})$ and $U_I + D_I = \Theta(n^{2 + \varepsilon})$. I.e., For square matrices, OPP has a lower asymptotic total communication cost than IPP.

\appendix

In this appendix we present the proofs to some of the results in the main text.

\subsection*{Proofs for Section V}

\begin{manuallemma}{1}
In the setting of Definition~5, it follows that,
\[
|L_i|=
\begin{cases}
\min\{L,2+\lfloor(T-1-i)/K\rfloor\} & \text{if } 1 \le i \le r \\
L & \text{if } r+1 \le i \le T ,
\end{cases}
\]
and,
\[
|R_i|=
\begin{cases}
\max\{0,K+T-KL-1\} & \text{if } i=1 \\
\max\{0,T-K+r-1\} & \text{if $2 \le i$ and $i \equiv 1 \pmod{r}$} \\
T-1 & \text{if } i \not\equiv 1 \pmod{r}
\end{cases}.
\]
\end{manuallemma}
\begin{proof}

First, note that $\bp$ and $\bs$ are chosen such that for any $K$ consecutive integers $a+[K-1]$ the sets 
\[ a+[K-1]+\SET(\bp)=a+[KL-1] \quad \text{and} \quad a+[K-1]+\SET(\bs)=a+KL+[K+T-2] \] are disjoint. Hence, for $\lambda r \le i \le (\lambda +1)r-1$, it follows that $L_i$ depends only on $L_j$ with $j \le \lambda r-1$ and in particular $|L_i| \le |(\as)_i + \SET(\bp)|=L$. Let $1 \le i \le r$.
Then, 
\begin{align*}
    L_i &= ((\as)_i + \SET(\bp)) \cap (\SET(\ap)+\SET(\bp|\bs)) \\
    &= (KL+i-1+K[L-1]) \cap [KL+K+T-2] .
\end{align*}
Thus, $x \in L_i$ if and only if $0 \le x \le KL+K+T-2$ and $x=KL+i-1+Kj$ with $0 \le j \le L-1$, i.e. \[KL+i-1+Kj \le KL+K+T-2, \quad \text{if and only if,} \quad j \le 1+(T-1-i)/K .\]
There are $2+\lfloor (T-1-i)/K \rfloor$ possible values for $j$.
Note, that since $i \le r \le T$, it follows that $0 \le 2+\lfloor (T-1-i)/K \rfloor$.

Let $r+1 \le i \le T$.
Then, 
\begin{align*}
    L_i &= (\as)_i + \SET(\bp) = (\as)_{i-r}+K+ \SET(\bp) \\
    &= (\as)_{i-r} + (\SET(\bp|\bs)_1 \setminus \{0\})
\end{align*} and $|L_i|=L$. 

Let $i=1$.
Then, 
\begin{align*}
    R_1 &= ((\as)_1 + \SET(\bs)) \cap (\SET(\ap)+\SET(\bp|\bs)) \\
    &= (2KL+[T-1]) \cap [KL+K+T-2] .
\end{align*}
Thus, $x \in R_1$ if and only if $0 \le x \le KL+K+T-2$ and $x=2KL+j$ with $0 \le j \le T-1$, i.e. \[ 2KL+j \le KL+K+T-2, \quad \text{if and only if,} \quad j \le K+T-KL-2 .\]
There are $\max\{0,K+T-KL-1\}$ possible values for $j$.
Note, that $K+T-KL-1 \le T$.

Let $2 \le i$ and $i \equiv 1 \pmod{r}$, i.e. $(\as)_i=KL+\lambda K$ for $1 \le \lambda$ and $(\as)_{i-1}=KL+(\lambda-1) K+r-1$.
Note, that since $\beta$ is sorted increasingly, $R_i$ does not depend on $\bp$ and since $\bs$ is an arithmetic progression, $R_i$ only depends on $R_{i-1}$.
Thus, 
\begin{align*}
    R_i &= ((\as)_i + \SET(\bs)) \cap ((\as)_{i-1}+\SET(\bs)) \\
    &= (2KL+\lambda K+[T-1])  \cap (2KL+(\lambda-1) K+r-1+[T-1]) \\
    &= 2KL+\lambda K+([T-1] \cap (-K+r-1+[T-1])) .
\end{align*}
So, $x \in R_i$ if and only if,
\[0 \le x-2KL-\lambda K = -K+r-1+j \le T-1 ,\] 
with $0 \le j \le T-1$.
Note that $-K+r-1+T-1 \le T-1$ and $0 \le K-r+1$.
Hence, $K-r+1 \le j \le T-1$ implies $|R_i| = \max\{0,T-K+r-1\}$.
Note, that $T-K+r-1 \le T$.

Let $i \not\equiv 1 \pmod{r}$.
Then using a similar argument as in the last paragraph, 
\begin{align*}
    R_i &= ((\as)_i + \SET(\bs)) \cap ((\as)_{i-1}+\SET(\bs)) \\
    &= ((\as)_{i-1}+1 + \SET(\bs)) \cap ((\as)_{i-1}+\SET(\bs)) \\
    &= (\as)_{i-1}+((1+\SET(\bs)) \cap \SET(\bs)) \\
    &= (\as)_{i-1}+KL+1+[T-2]
\end{align*} of size $T-1$.

\end{proof}

\begin{manualtheorem}{1}
The number of distinct entries in the degree table of $\mathsf{GASP}_{r}$ is given by
\begin{multline*}
  N= KL+2K+3T-2 -\max\{K,\varphi\}+(L-2)\max\{0,\min\{r,r-\varphi\}\}+ \lfloor (T-1)/r \rfloor \min\{T-1,K-r\} \\ -\mathbf{1}_{\varphi < r}\Bigg(\min\{0,\mu-r\} 
  +r(T-1-\mu)/K+\frac{-Kx^2 + (-K-2\max\{0,\varphi\}+2T-2)x+T-1-\mu}{2} \\ -\frac{T-1-\mu}{K} \cdot \frac{T-1+\mu}{2}\Bigg),
\end{multline*}
where $\varphi = T-1-KL+2K$, $\mu \equiv T-1 \pmod{K}$ with $0 \le \mu \le K-1$, and $x = \min\left\{ \frac{T-1-\mu}{K} -\mathbf{1}_{\mu=0}, L-3 \right\}$.
\end{manualtheorem}

\begin{proof}
We compute $S = \sum_{i=1}^{T} |L_i| + \sum_{i=1}^{T} |R_i|$ in Equation~5 using Lemma~1.

First, we focus on $\sum_{i=1}^{T} |L_i|$.
By the definition of $\varphi$ we have $i \le \varphi$ if and only if $L \le 2+\lfloor(T-1-i)/K\rfloor$. Let $z = \max\{1,\varphi+1\}$, then,
\begin{align*}
\sum_{i=1}^{T} |L_i|
&= \sum_{i=1}^{\min\{r,\varphi\}} L + \sum_{i=z}^{r} |L_i| + \sum_{i=r+1}^{T} L
= \max\{0,\min\{r,\varphi\}\} L + \sum_{i=z}^{r} |L_i| + (T-r) L,
\end{align*}
and,
\begin{align*}
\sum_{i=z}^{r} |L_i|
&= \sum_{i=z}^{r} (2+\lfloor(T-1-i)/K\rfloor)
= 2\max\{0,r-z+1\} + \sum_{i=T-1-r}^{T-1-z} \lfloor i/K\rfloor.
\end{align*}
Note that, $ \max\{0,r-z+1\}
= \max\{0,\min\{r,r-\varphi\}\}
= r-\max\{0,\min\{r,\varphi\}\}$.

Writing $T-1-r = aK+b$ and $T-1-z = xK+y$ with $0 \le b \le K-1$ and $0 \le y \le K-1$, we get $0$ if $r < z$ and else
\begin{align*}
\sum_{i=aK+b}^{xK+y} \lfloor i/K\rfloor
&= (K-b)a+K\sum_{i=a+1}^{x-1}i+(y+1)x
= K(x-a)(x+a-1)/2-ab+xy+x.
\end{align*}
Note, that $r < z$ if and only if $r \le \varphi$.
Using the definition of $\mu$, we simplify,
\begin{align*}
a
=\lfloor (T-1-r)/K \rfloor
=(T-1-\mu)/K-\mathbf{1}_{\mu < r} ,
\end{align*}
and,
\begin{align*}
x
&= \lfloor (T-1-z)/K \rfloor
= \lfloor (T-1-\max\{1,\varphi+1\})/K \rfloor
= \left\lfloor \min\left\{\frac{T-2}{K},\frac{KL-2K-1}{K} \right\} \right\rfloor
\\ &= \min\left\{ \frac{T-1-\mu}{K} -\mathbf{1}_{\mu=0}, L-3 \right\}.
\end{align*}
Now, we plug $b=T-1-r -aK$, $y=T-1-z -xK$ and $a=(T-1-\mu)/K-\mathbf{1}_{\mu < r}$ in:
\begin{align*}
&K(x-a)(x+a-1)/2-ab+xy+x
= \min\{0,\mu-r\}+r\left(\frac{T-1-\mu}{K}\right)
\\&+\frac{-Kx^2 + (-K-2\max\{0,\varphi\}+2T-2)x+T-1-\mu}{2}
\\&-\frac{T-1-\mu}{K} \cdot \frac{T-1+\mu}{2}.
\end{align*}

Second, we compute $\sum_{i=1}^{T} |R_i|$ and define $\eta$ as
\begin{align*}
\eta
&
=
|\{i \mid 2 \le i \le T, i \equiv 1 \pmod{r}\}|
=
|\{\lambda \mid 2 \le 1+\lambda r \le T\}|
=
|\{\lambda \mid \lceil 1/r \rceil \le \lambda \le \lfloor (T-1)/r \rfloor \}|
\\
&
=
\lfloor (T-1)/r \rfloor.
\end{align*}
Note that $0 \le \eta$.
Hence, since $|R_1| = \max\{K,\varphi\}-K$ and $|R_{r+1}| -T+1 = -\min\{T-1,K-r\}$, it follows that,
\begin{align*}
\sum_{i=1}^{T} |R_i|
&= |R_1|
+ \sum_{\substack{i=2 \\ i \equiv 1 \pmod{r}}}^{T} |R_{r+1}|
+ \sum_{\substack{i=1 \\ i \not\equiv 1 \pmod{r}}}^{T} (T-1)
= |R_1| + \eta |R_{r+1}| + (T-1-\eta) (T-1)
\\&= |R_1| + \eta (|R_{r+1}|-T+1) + (T-1)^2
\\&= \max\{K,\varphi\}-K - \eta \min\{T-1,K-r\} + (T-1)^2.
\end{align*}

Plugging all terms in Equation~5 yields
\begin{align*}
S &= \max\{0,\min\{r,\varphi\}\} (L-2) + (T-r) L + 2r
\\&+\max\{K,\varphi\}-K - \eta \min\{T-1,K-r\} + (T-1)^2
\\&+\mathbf{1}_{\varphi < r}(\min\{0,\mu-r\}+r\left(\frac{T-1-\mu}{K}\right)
\\&+\frac{-Kx^2 + (-K-2\max\{0,\varphi\}+2T-2)x+T-1-\mu}{2}
\\&-\frac{T-1-\mu}{K} \cdot \frac{T-1+\mu}{2})
\end{align*}
and Equation~6 finishes the proof.
\end{proof}

\subsection*{Proofs for Section VI}

\begin{manualproposition}{1}
For $K=L=T=n^2$ the optimal chain length is given by $r^{*} = n$, with 
\begin{align} \label{eq:nsquare}
N = \left\{\begin{matrix}
3 & \text{if} \quad n=1 \\ 
n^4+2n^3+2n^2-n-2 & \text{if} \quad n \geq 2
\end{matrix}\right. .
\end{align}
\end{manualproposition}

\begin{proof}
If $n=1$, then $1 \le r \le \min\{K,T\}$ implies $r^{*}=1$ and we compute $N(1)=3$. 

We now consider the case for $n \geq 2$.
Using the notation of Theorem~1, we compute $\varphi = -n^4+3n^2-1 \le -5$, $\mu = n^2-1$, and $x = 0$, so the formula for $N(r)$ simplifies to
\begin{align*}
N(r) =
n^4+4n^2-2+(n^2-2)r
+ \lfloor (n^2-1)/r \rfloor (n^2-r)+\mathbf{1}_{r=n^2}.
\end{align*}
Note, that $N(n) = n^4+2n^3+2n^2-n-2$ and $N(n^2) = 2n^4+2n^2-1$. Thus, $N(n) < N(n^2)$ and we can assume $r \le n^2-1$, obtaining $N(r) = n^4+4n^2-2+(n^2-2)r + \lfloor (n^2-1)/r \rfloor (n^2-r)$.
It follows from $\lfloor x \rfloor > x-1$ that,
\begin{align} \label{eq:lemma 1 proof 1}
N(r) > n^4+4n^2-2+(n^2-2)r + (n^2-1-r)(n^2-r)/r.
\end{align}
The right hand side of \eqref{eq:lemma 1 proof 1} is greater or equal than $N(n)$ if and only if $0 \le (n+1-r)(n^2-(n+1)r)$, i.e. if $r \le \lfloor \frac{n^2}{n+1} \rfloor = n-1$ or $n+1 \le r$, concluding the proof.
\end{proof}

\begin{manuallemma}{2}
In the setting of Theorem~1, if $r \le \varphi$, then $r^{*}=\min\{K,T,\varphi\}$ minimizes $N(r)$.
\end{manuallemma}

\begin{proof}
If $r \le \varphi$, then the formula of $N$ in Theorem~1 simplifies to
\begin{align}
\label{eq_N_rlephi}
\begin{split}
N(r) =
&KL+2K+3T-2 -\max\{K,\varphi\} + \lfloor (T-1)/r \rfloor \min\{T-1,K-r\}.
\end{split}
\end{align}
So, regarded as function in $r \in \{1,\ldots,\min\{K,T\}\}$, both $r$-dependent factors are monotonically decreasing and attain their minimum for the largest $r$, i.e. $r^{*}=\min\{K,T,\varphi\}$.
\end{proof}

\subsection*{Proofs for Section VII}

\begin{manuallemma}{4}
Let $(\alpha,\beta) \in \A(K,L,T)$ be a degree table and $p=\pi \circ \alpha$, $q=\tau \circ \beta$ the corresponding sorted vectors for unique sorting permutations $\pi$ and $\tau$.
Moreover, we abbreviate $a=\min\{\SET(\alpha)\}$, $A=\max\{\SET(\alpha)\}$, $b=\min\{\SET(\beta)\}$, and $B=\max\{\SET(\beta)\}$.

\begin{enumerate}
\item For an integer $i$ with $1 \le i < K+T$ and $p_i+B < (p_{i+1}-1)+b$, there is a degree table $(\alpha',\beta) \in \A(K,L,T)$ with $\N(\alpha,\beta) = \N(\alpha',\beta)$ where
$
\alpha'_j
=
\begin{cases}
\alpha_j & \text{if } j \le \pi(i) \\
\alpha_j-1 & \text{else}
\end{cases}
$.
\item For an integer $i$ with $1 \le i < L+T$ and $q_i+A < (q_{i+1}-1)+a$, there is a degree table $(\alpha,\beta') \in \A(K,L,T)$ with $\N(\alpha,\beta) = \N(\alpha,\beta')$ where
$
\beta'_j
=
\begin{cases}
\beta_j & \text{if } j \le \tau(i) \\
\beta_j-1 & \text{else}
\end{cases}
$.
\end{enumerate}
\end{manuallemma}

\begin{proof}
We prove (1.) and the argumentation for (2.) is essentially the same.
Fix an integer $1 \le i < K+T$ with $p_i+B < (p_{i+1}-1)+b$ and define the new degree table via $(\alpha',\beta)$.
Since $(\alpha,\beta)$ is a degree table, all entries in $\beta$ are distinct and using $B-b < (p_{i+1}-1)-p_i$, all integers in $\alpha'$ are distinct.

Now we show that $\alpha_{m}+\beta_{n} = \alpha_{m'}+\beta_{n'}$ iff $\alpha_{m}'+\beta_{n} = \alpha_{m'}'+\beta_{n'}$ so that the contraposition $\alpha_{m}+\beta_{n} \ne \alpha_{m'}+\beta_{n'}$ iff $\alpha_{m}'+\beta_{n} \ne \alpha_{m'}'+\beta_{n'}$, which then completes the proof.

If $\alpha_{m}'=\alpha_{m}$ and $\alpha_{m'}' = \alpha_{m'}$, i.e. $m,m' \le \pi(i)$, or $\alpha_{m}'=\alpha_{m}-1$ and $\alpha_{m'}'=\alpha_{m'}-1$, i.e. $m,m' > \pi(i)$, this is obviously true.
Hence, we assume wlog. $\alpha_{m}'=\alpha_{m}$ and $\alpha_{m'}' = \alpha_{m'}-1$, i.e. $m \le \pi(i) < m'$.

Then, using $p_i+B < (p_{i+1}-1)+b$ which is $\alpha_{\pi(i)}+B < (\alpha_{\pi(i+1)}-1)+b$, we get
\begin{align*}
&\alpha_{m}+\beta_{n}
=
\alpha_{m}'+\beta_{n}
\le
\alpha_{\pi(i)}+B
\\&<
(\alpha_{\pi(i+1)}-1)+b
\le
\alpha_{m'}'+\beta_{n'}
=
\alpha_{m'}-1+\beta_{n'},
\end{align*}
i.e. $\alpha_{m}+\beta_{n} < \alpha_{m'}+\beta_{n'}$ and $\alpha_{m}'+\beta_{n} < \alpha_{m'}'+\beta_{n'}$.
This completes the proof.
\end{proof}

\begin{manuallemma}{5}
Let $(\alpha,\beta)$ be a degree table.
Then at most one of the two operations (1.) and (2.) in Lemma~4 are applicable for $(\alpha,\beta)$.
Two successive operations of (1.) respectively (2.) using the indices $i,i'$ applied to $(\alpha,\beta)$ yield the same degree table as if the indices $i',i$ were used.
\end{manuallemma}

\begin{proof}
We use the notation of Lemma~4. Assume that (1.) and (2.) are both applicable, then there are $i$ and $j$ such that $p_i+B < (p_{i+1}-1)+b$ and $q_i+A < (q_{i+1}-1)+a$, but $B-b+1 < p_{i+1}-p_i \le A-a$ and $A-a < q_{i+1}-q_i-1 \le B-b-1$ is a contradiction.

Assume two successive operations of (1.) with $i$ and $i'$.
The case of (2.) is similar.

After performing the both operations in any order, the resulting degree table is $(\alpha',\beta)$ with 
\[
\alpha'_j
=
\begin{cases}
\alpha_j & \text{if } j \le \min\{\pi(i),\pi(i')\} \\
\alpha_j-1 & \text{if } \min\{\pi(i),\pi(i')\} < j \le \max\{\pi(i),\pi(i')\} \\
\alpha_j-2 & \text{if } \max\{\pi(i),\pi(i')\} < j
\end{cases},
\]
showing that it is order independent.
\end{proof}

\begin{manuallemma}{7}
The relation defined in Definition~6 is an equivalence relation and we have $\N(\ap,\as,\bp,\bs)=\N(\ap',\as',\bp',\bs')$ if $(\ap,\as,\bp,\bs) \sim (\ap',\as',\bp',\bs')$.
\end{manuallemma}

\begin{proof}
Denote $A=(\ap,\as,\bp,\bs)$, $A'=(\ap',\as',\bp',\bs')$, and $A''=(\ap'',\as'',\bp'',\bs'')$.
The relation is reflexive ($A=1 \cdot (id \circ A + (0,0))$), symmetrical ($A = r \cdot (\sigma \circ A' + (s,t)) \Leftrightarrow A' = r^{-1} \cdot (\sigma^{-1} \circ A + (-s,-t))$, here, $\sigma^{-1}$ means component wise inversion), and transitive: if $A = r \cdot (\sigma \circ A' + (s,t))$ and $A' = r' \cdot (\sigma' \circ A'' + (s',t'))$, then $A = rr' \cdot (\sigma\sigma' \circ A'' + (s'+s/r',t'+t/r'))$ ($\sigma\sigma'$ is again component wise).

Assume, we have $(\alpha,\beta) \sim (\alpha',\beta') = r \cdot (\sigma \circ (\alpha,\beta) + (s,t))$.
Then $\alpha_i+\beta_j \in \SET(\alpha)+\SET(\beta)$ iff $r(\alpha_{(\kappa,\lambda)((\kappa,\lambda)^{-1}(i))}+s)+r(\beta_{(\mu,\nu)((\mu,\nu)^{-1}(j))}+t) \in \SET(r \cdot ((\kappa,\lambda)\circ\alpha+s\mathds{1}_{K+T}))+\SET(r \cdot ((\mu,\nu)\circ\beta+t\mathds{1}_{L+T}))$ and 
\begin{align*}
&\alpha_i+\beta_j = \alpha_{i'}+\beta_{j'}
\\
&\Leftrightarrow
\\&r(\alpha_{(\kappa,\lambda)((\kappa,\lambda)^{-1}(i))}+s)+r(\beta_{(\mu,\nu)((\mu,\nu)^{-1}(j))}+t)
\\&= r(\alpha_{(\kappa,\lambda)((\kappa,\lambda)^{-1}(i'))}+s)+r(\beta_{(\mu,\nu)((\mu,\nu)^{-1}j'))}+t).
\end{align*}
Hence, $\N(\alpha,\beta) = \N(\alpha',\beta')$.
\end{proof}

\subsection*{Proofs for Section VIII}

\begin{manuallemma}{8}[{[23],[24]}]

Let $A$ and $B$ be sets of integers.
Then $|A|+|B|-1 \le |A+B|$ and if $2 \le |A|, |B|$, then equality holds if and only if $A$ and $B$ are arithmetic progressions with the same common difference, i.e. $A=a+d[m]$ and $B=b+d[n]$ for $a,b,d,m,n \in \mathbb{Z}$.

\end{manuallemma}

\begin{proof}
Let $A=\{a_0,\ldots,a_m\}$ and $B=\{b_0,\ldots,b_n\}$ with $a_i < a_{i+1}$ and $b_j < b_{j+1}$.
Consider $M=\{a_0+b_0,a_0+b_1,\ldots,a_0+b_n,a_1+b_n,\ldots,a_m+b_n\}$ of size $|A|+|B|-1$.
$M \subseteq A+B$ proves the bound.
Assume now that $|A|+|B|-1 = |A+B|$.
Consider additionally $N=\{a_0+b_0,a_1+b_0,a_1+b_1,\ldots,a_1+b_{n-1},a_2+b_{n-1},\ldots,a_m+b_{n-1},a_m+b_n\}$ of size $|A|+|B|-1$.
Then, $A+B=M=N$ showing $a_0+b_i=a_1+b_{i-1}$ for $1 \le i \le n$ and $a_{j-1}+b_n=a_j+b_{n-1}$ for $1 \le j \le m$, i.e. $d:=a_1-a_0=b_i-b_{i-1}=b_n-b_{n-1}=a_j-a_{j-1}$, completing the proof.
\end{proof}

%
%
%

\begin{manuallemma}{9}
Let $K$, $L$, and $T$ be positive integers with $L \le K$ and $(\ap,\as,\bp,\bs)$ be a degree table in $\A(K,L,T)$ with $\N(\ap,\as,\bp,\bs) = KL+K+2T-1$.
Then there are nonnegative integers $a$, $b$, $d$, $x$, $y$, $e$ such that $1 \le d$, $n=T-1$, $m=K+T-1$ and
\begin{enumerate}
\item $\SET(\ap|\as) = \SET(\ap) \dot\cup \SET(\as) = a+d[m]$,
\item $\SET(\bs) = b+d[n]$,
\item $\SET(\as)+\SET(\bp) \subseteq \SET(\ap|\as)+\SET(\bs) = a+b+d[m+n]$,
\item $\SET(\bp) \subseteq b+d(\{-K,\ldots,-1\} \dot\cup \{n+1,\ldots,m+n\})$,
\item $\SET(\ap)+\SET(\bp) \subseteq \SET(\ap|\as)+\SET(\bp) = a+b+d\{-K,\ldots,2m+n\}$,
\item $|\SET(\ap)+\SET(\bp)| \le |a+b+d\{-K,\ldots,2m+n\}| = 2m+n+K+1$, and
\item $KL \le 3K+3T-2$, i.e. $L \le 3+\frac{3T-2}{K}$.
\end{enumerate}
If additionally $K=L$, then
\begin{enumerate}
\setcounter{enumi}{7}
\item $\SET(\as)=x+e[n]$ and $\SET(\beta)=y+e[m]$,
\item $\SET(\alpha)+\SET(\bs) = \SET(\as)+\SET(\beta) = a+b+d[m+n] = x+y+e[m+n]$, i.e. $d=e$,
\item $\SET(\alpha)+\SET(\beta) = a+y+d[2m]$, and
\item $K=L=1$.
\end{enumerate}
\end{manuallemma}

\begin{proof}
(1.) -- (3.)

If the bound in Theorem~4 is tight and $L\le K$, the set $\SET(\as)+\SET(\bp)$ must be subset of $\SET(\alpha)+\SET(\bp)$, else it would increase the bound, and the second part of Lemma~8 shows that $\SET(\alpha)$ and $\SET(\bs)$ are arithmetic progressions with the same common difference.

(4.)

By (3.) we write $a+d\lambda \in \SET(\as)$, $\lambda \in [m]$, $x \in \SET(\bp)$, and $a+b+d\mu \in a+b+d[m+n]$.
Solving for $x$ yields $x=b+d(\mu-\lambda)$ with $-m \le \mu-\lambda \le m+n$.
Since $\SET(\bp) \cap \SET(\bs) = \emptyset$, we have $\mu-\lambda \not \in [n]$.
Since $|\SET(\as)|=T$, there is a $a+d\lambda' \in \SET(\as)$ with $\lambda' \le K$.
If $\eta \le -K-1$, then $(a+d\lambda')+(b+d\eta) \le a+b+d(K-K-1) < 0$ and $(a+d\lambda')+(b+d\eta) \not \in a+b+d[m+n]$.

(5.)
\begin{align*}
&
\SET(\ap)+\SET(\bp)
\subseteq
\SET(\ap|\as)+\SET(\bp)
\\
&
=
a+b+d([m]+(\{-K,\ldots,-1\} \dot\cup \{n+1,\ldots,m+n\}))
\\
&
=
a+b+d(\{-K,\ldots,m-1\} \cup \{n+1,\ldots,2m+n\})
\\
&
=
a+b+d\{-K,\ldots,2m+n\}
\end{align*}

(6.) follows directly from (5.) and (7.) follows directly from (6.).

(8.)

If $K=L$, then the argumentation for bullet point (1.) can be applied for $K$ and $L$ with interchanged roles.

(9.) follows directly and $d=e$ since due to $m+n=K+2T-2 \ge 1$ both sets contain at least two elements.

(11.) follows by using the size $|\SET(\alpha)+\SET(\beta)|$ from (10.) (i.e. $2K+2T-1$) and because the bound is tight $\N(\ap,\as,\bp,\bs) = KL+K+2T-1$, implying $0=K(K-1)$.
\end{proof}

\begin{manuallemma}{10}
Let $(\ap,\as,\bp,\bs) \in \A(K,L,T)$. Then, for all $1 \le i \le K+T$ and all $1 \le j \le L+T$, it holds that $|(\alpha_i+\SET(\beta_p)) \cap (\SET(\alpha_p)+\beta_j)| \le 1$.
\end{manuallemma}

\begin{proof}
We abbreviate $M_{i,j}=(\alpha_i+\SET(\beta_p)) \cap (\SET(\alpha_p)+\beta_j)$.

If either $i \le K$ or $j \le L$, then $M_{i,j} = \emptyset$ due to Property~(3) of Definition~3.

If $i \le K$ and $j \le L$, then $M_{i,j} = \{\alpha_i+\beta_j\}$ due to Property~(3) of Definition~3.

Else, i.e. $K+1 \le i$ and $L+1 \le j$, we assume that $|M_{i,j}| \ge 2$.
Let $x=\alpha_{i}+\beta_{m}=\alpha_{n}+\beta_{j}$ and $y=\alpha_{i}+\beta_{o}=\alpha_{p}+\beta_{j}$ be two distinct elements in the intersection for $(i,m) \ne (n,j)$ and $(i,o) \ne (p,j)$ and $m,o \le L$ and $n,p \le K$.
Then we have $\alpha_{i}-\beta_{j}=\alpha_{n}-\beta_{m}=\alpha_{p}-\beta_{o}$ and $\alpha_{n}+\beta_{o}=\alpha_{p}+\beta_{m}$.
Using Property~(3) of Definition~3, this implies $(n,o)=(p,m)$, which is a contradiction to $x \ne y$.
\end{proof}

\begin{manualtheorem}{5}
Let $(\alpha,\beta) \in \A(K,L,T)$. Then, $KL+K+L+2T-1-T\min\{K,L,T\} \le \N(\alpha,\beta)$. Thus, for $\ap=(0,1,\ldots,K-1)$, $\as=(KL)$, $\bp=(0,K,\ldots,K(L-1))$, and $\bs=(KL)$, the degree table $(\ap,\as,\bp,\bs) \in \A(K,L,1)$ is of minimum size.
\end{manualtheorem}

\begin{proof}
Let $(\ap,\as,\bp,\bs) \in \A(K,L,T)$ of minimum size and $\alpha=\ap|\as$, $\beta=\bp|\bs$.
For the first part, we use Lemma~8 in
\begin{align*}
&
|(\SET(\as)+\SET(\bp)) \cup (\SET(\alpha)+\SET(\bs))|
\\
\ge
&
|(\alpha_{K+1}+\SET(\bp)) \cup (\SET(\alpha)+\SET(\bs))|
\\
=
&
|\alpha_{K+1}+\SET(\bp)| + |\SET(\alpha)+\SET(\bs)|
\\&- |(\alpha_{K+1}+\SET(\bp)) \cap (\SET(\alpha)+\SET(\bs))|
\\
\ge
&
L + (K+2T-1) - |(\alpha_{K+1}+\SET(\bp)) \cap (\SET(\alpha)+\SET(\bs))|
\end{align*}
and rewrite
\begin{align*}
&
(\alpha_{K+1}+\SET(\bp)) \cap (\SET(\alpha)+\SET(\bs))
\\
=
&
(\alpha_{K+1}+\SET(\bp)) \cap \left(\bigcup_{i=L+1}^{L+T} (\SET(\alpha)+\beta_i) \right)
\\
=
&
(\alpha_{K+1}+\SET(\bp)) \cap \left(\bigcup_{i=L+1}^{L+T} ((\SET(\ap)+\beta_i) \dot\cup (\SET(\as)+\beta_i)) \right)
\\
=
&
\bigcup_{i=L+1}^{L+T} (((\alpha_{K+1}+\SET(\bp)) \cap (\SET(\ap)+\beta_i))
\\&\dot\cup ((\alpha_{K+1}+\SET(\bp)) \cap (\SET(\as)+\beta_i)))
\end{align*}
so that
\begin{align*}
&
\Bigg| \bigcup_{i=L+1}^{L+T} (((\alpha_{K+1}+\SET(\bp)) \cap (\SET(\ap)+\beta_i))
\\&\dot\cup ((\alpha_{K+1}+\SET(\bp)) \cap (\SET(\as)+\beta_i))) \Bigg|
\\
\le
&
\sum_{i=L+1}^{L+T} (|(\alpha_{K+1}+\SET(\bp)) \cap (\SET(\ap)+\beta_i)|
\\&+ |(\alpha_{K+1}+\SET(\bp)) \cap (\SET(\as)+\beta_i)|).
\end{align*}
Lemma~10 shows now
\begin{align*}
|(\alpha_{K+1}+\SET(\bp)) \cap (\SET(\ap)+\beta_i)| = z_i \in \{0,1\}
\end{align*}
and we have
\begin{align*}
|(\alpha_{K+1}+\SET(\bp)) \cap (\SET(\as)+\beta_i)| \le T-1
\end{align*}
since $|\SET(\as)+\beta_i|=T$ and $\alpha_{K+1}+\beta_i \not\in \alpha_{K+1}+\SET(\bp)$ by Property~(2) of Definition~3 and
\begin{align*}
|(\alpha_{K+1}+\SET(\bp)) \cap (\SET(\as)+\beta_i)| \le L-z_i
\end{align*}
since $|\alpha_{K+1}+\SET(\bp)|=L$ and if $z_i=1$, then there are $1 \le x \le L$ and $1 \le y \le K$ such that $\alpha_{K+1}+\beta_x = \alpha_y+\beta_i$, but $\alpha_{K+1}+\beta_x \in \alpha_{K+1}+\SET(\bp)$ and $\alpha_y+\beta_i \not\in \SET(\as)+\beta_i$ by Property~(1) of Definition~3.

Hence, we continue
\begin{align*}
&
\sum_{i=L+1}^{L+T} (|(\alpha_{K+1}+\SET(\bp)) \cap (\SET(\ap)+\beta_i)|
\\&+ |(\alpha_{K+1}+\SET(\bp)) \cap (\SET(\as)+\beta_i)|)
\\
\le
&
\sum_{i=L+1}^{L+T} (z_i + \min\{T-1,L-z_i\})
\\
=
&
\sum_{i=L+1}^{L+T} \min\{T-1+z_i,L\}
\\
\le
&
\sum_{i=L+1}^{L+T} \min\{T,L\}
=
T\min\{T,L\}.
\end{align*}

So a similar argument as in the proof of Theorem~4 shows
\begin{align*}
\N(K,L,T)
&
=
\N(\ap,\as,\bp,\bs)
\\&=
KL+|(\SET(\as)+\SET(\bp)) \cup (\SET(\alpha)+\SET(\bs))|
\\
&
\ge
KL + L + (K+2T-1) - (T\min\{T,L\}).
\end{align*}

Exchanging $K$ and $L$ and virtually the same argumentation shows
\begin{align*}
\N(K,L,T)
\ge
KL + K + (L+2T-1) - (T\min\{T,K\}).
\end{align*}

So that
\begin{align*}
\N(K,L,T)
&
\ge
KL+K+L+2T-1
\\&+\max\{-T\min\{T,L\},-T\min\{T,K\}\}
\\
&
=
KL+K+L+2T-1
\\&-T\min\{\min\{T,L\},\min\{T,K\}\}
\\
&
=
KL+K+L+2T-1-T\min\{K,L,T\}.
\end{align*}

Last, fix $\ap=(0,1,\ldots,K-1)$, $\as=(KL)$, $\bp=(0,K,\ldots,K(L-1))$, and $\bs=(KL)$.
Then, we have $\SET(\alpha)+\SET(\beta) = [KL+K-1] \dot\cup (KL+K+K[L-1])$ of size $KL+K+L$.
\end{proof}

\bibliographystyle{IEEEtran}
\bibliography{ref.bib}

\end{document}